%% file: main-model-theory-KR25.tex
\documentclass{article}
\usepackage{kr}
\usepackage{times}
\usepackage{soul}
\usepackage{url}
\usepackage[hidelinks]{hyperref}
\usepackage[utf8]{inputenc}
\usepackage[small]{caption}
\usepackage{graphicx}
\usepackage{amsmath}
\usepackage{amsthm}
\usepackage{booktabs}
\usepackage{comment}
\usepackage{algpseudocode}
\usepackage[ruled,vlined]{algorithm2e}

\SetKwRepeat{Repeat}{repeat}{until}
\SetKwFor{UnivSel}{universally select}{do}{end}
\SetKwFor{UnivDo}{universally}{do:}{end}
\SetKwFor{ForAll}{for each}{do}{end}

\input{macros.tex}

\newtheorem{example}{Example}
\newtheorem{theorem}{Theorem}

\newtheorem{proposition}[theorem]{Proposition}
\newtheorem{lemma}[theorem]{Lemma}
\newtheorem{claim}[theorem]{Claim}
\newtheorem{definitionAux}[theorem]{Definition}
\newenvironment{definition}{\begin{definitionAux}
	}{\end{definitionAux}}

\title{Finite Axiomatizability by Disjunctive Existential Rules}

\author{
	Marco Calautti$^1$\and
	Marco Console$^2$\and
	Andreas Pieris$^{3,4}$\\
	\affiliations
	$^1$University of Milan\\
	$^2$Sapienza University of Rome\\
	$^3$University of Cyprus\\
	$^4$University of Edinburgh\\[1mm]
	\emails 
	marco.calautti{@}unimi.it, console{@}diag.uniroma.it, apieris{@}inf.ed.ac.uk
}

\begin{document}

\maketitle

\begin{abstract}
Rule-based languages lie at the core of several areas of central importance to databases
and artificial intelligence such as deductive databases and knowledge representation and reasoning. Disjunctive existential rules (a.k.a. disjunctive tuple-generating dependencies in the database literature) form such a prominent rule-based language. The goal of this work is to pinpoint the expressive power of disjunctive existential rules in terms of insightful model-theoretic properties. More precisely, given a collection $\coll{C}$ of relational structures, we show that $\coll{C}$ is axiomatizable via a finite set $\dep$ of disjunctive existential rules (i.e., $\coll{C}$ is precisely the set of models of $\dep$) iff $\coll{C}$ enjoys certain model-theoretic properties. This is achieved by using the well-known property of criticality, a refined version of closure under direct products, and a novel property called diagrammatic compatibility that relies on the method of diagrams. 
We further establish analogous characterizations for the well-behaved classes of linear and guarded disjunctive existential rules by adopting refined versions of diagrammatic compatibility that consider the syntactic restrictions imposed by linearity and guardedness; this illustrates the robustness of diagrammatic compatibility. 
We finally exploit diagrammatic compatibility to rewrite a set of guarded disjunctive existential rules into an equivalent set that falls in the weaker class of linear disjunctive existential rules, if one exists.
\end{abstract}

\input{introduction.tex}
\input{preliminaries.tex}
\input{properties.tex}
\input{ders.tex}
\input{linear-ders.tex}
\input{rewritability.tex}
\input{conclusions.tex}

\bibliographystyle{kr}
\bibliography{references}

\newpage
\appendix
\input{appendix.tex}

\end{document}

%% file: macros.tex
\newcommand{\ins}[1]{\mathbf{#1}}
\newcommand{\adom}[1]{\mathsf{dom}(#1)}
\newcommand{\aadom}[1]{\mathsf{adom}(#1)}

\newcommand{\coll}[1]{\mathcal{#1}}

\renewcommand{\paragraph}[1]{\textbf{#1}}
\newcommand{\ra}{\rightarrow}

\newcommand{\dep}{\Sigma}

\newcommand{\ar}[1]{\mathsf{ar}(#1)}
\newcommand{\body}[1]{\mathsf{body}(#1)}
\newcommand{\head}[1]{\mathsf{head}(#1)}

\newcommand{\class}[1]{\mathsf{#1}}

\newcommand{\result}[1]{\mathsf{result}(#1)}
\newcommand{\paths}[1]{\mathsf{paths}(#1)}
\newcommand{\chase}{\mathsf{chase}}

\newcommand{\facts}[1]{\mathsf{facts}(#1)}

\def\qed{\hfill{\qedboxempty}      % qed with empty box
  \ifdim\lastskip<\medskipamount \removelastskip\penalty55\medskip\fi}

\def\qedboxempty{\vbox{\hrule\hbox{\vrule\kern3pt
                 \vbox{\kern3pt\kern3pt}\kern3pt\vrule}\hrule}}

\def\qedfull{\hfill{\qedboxfull}   % qed with full box
  \ifdim\lastskip<\medskipamount \removelastskip\penalty55\medskip\fi}

\def\qedboxfull{\vrule height 4pt width 4pt depth 0pt}

\newcommand{\markfull}{\qedboxfull}

%% file: introduction.tex
\section{Introduction}\label{sec:introduction}

Model theory is the study of the interaction between sentences in some logical formalism and their models, that is, structures that satisfy the sentences. There are two directions in this interaction: from syntax to semantics and from semantics to syntax.
The first direction aims to identify properties that are satisfied by all models of sentences having common syntactic features, i.e., the syntax of a logical formalism is introduced first and then the properties of the mathematical structures that satisfy
sentences of that formalism are explored. The second direction aims to characterize sentences in terms of their model-theoretic properties, i.e., given a set of properties, the goal is to determine whether or not these properties characterize some logical formalism.
In general, establishing results in the second direction is much harder than establishing results in the first. In other words, obtaining model-theoretic characterizations of logical formalisms is a far greater challenge than identifying properties enjoyed by all models of sentences from a certain formalism.

In this work, we are concerned with the second direction of the interaction between logical sentences and their models. More precisely, we are interested in model-theoretic results of the following form. For a logical formalism $\coll{L}$, given a collection $\coll{C}$ of structures, the following are equivalent:
\begin{enumerate}
	\item There exists a finite set $\Phi$ of sentences from $\coll{L}$ such that $\coll{C}$ is precisely the set of models of $\Phi$, in which case we say that $\coll{C}$ is {\em finitely axiomatizable by $\coll{L}$}.\footnote{Note that one can also talk about the notion of axiomatizability by $\coll{L}$, where the set of sentences $\Phi$ from $\coll{L}$ might be infinite.}
	\item $\coll{C}$ enjoys certain model-theoretic properties.
\end{enumerate}
Such results are welcome as they pinpoint the absolute expressive power of a logical formalism in terms of insightful model-theoretic properties. In other words, they equip us with model-theoretic tools that allow us to prove or disprove that a set of sentences $\Phi$ from some logical formalism can be equivalently rewritten as a finite set of sentences from $\coll{L}$. Indeed, by showing that the collection of models of $\Phi$ enjoy the properties listed in item (2), then $\Phi$ is logically equivalent to a finite set $\Phi'$ of sentences from $\coll{L}$; otherwise, we can safely conclude that such $\Phi'$ does not exist.

Two classical results in model theory of the above form are the characterizations of when a collection $\coll{C}$ of structures is (finitely) axiomatizable by first-order sentences via the property of closure under isomorphisms, and the more sophisticated closure properties under ultraproducts and ultrapower factors; the definitions of those notions are not essential for our discussion. 
In particular, we know that a collection $\coll{C}$ of structures over a schema (a.k.a.~signature) $\ins{S}$ is axiomatizable by first-order sentences iff $\coll{C}$ is closed under isomorphisms, ultraproducts, and ultrapower factors.
Moreover, we know that $\coll{C}$ is finitely axiomatizable by first-order sentences iff both $\coll{C}$ and the complement of $\coll{C}$, which collects all the structures over $\ins{S}$ that are not in $\coll{C}$, are closed under isomorphisms, ultraproducts, and ultrapower factors.

\subsection{Characterizations for Rule-based Formalisms}

Interesting characterizations have been also established for rule-based formalisms, expressed in suitable fragments of first-order logic, that have been used in numerous areas of central importance to databases and knowledge representation and reasoning. A prominent such formalism, which has attracted considerable attention during the last decades, is that of existential rules (a.k.a.~tuple-generating dependencies in the database literature) that are first-order sentences
\[
\forall \bar x \forall \bar y \left(\phi(\bar x,\bar y)\ \ra\ \exists \bar z\, \psi(\bar x,\bar z)\right), 
\]
where $\phi$ and $\psi$ are conjunctions of atoms.
Existential rules have been originally introduced as a framework for database constraints~\cite{AbHV95}.
Later on, they have been deployed in the study of data exchange~\cite{DBLP:journals/tcs/FaginKMP05} and data integration \cite{DBLP:conf/pods/Lenzerini02}. More recently, they have been used for knowledge representation purposes, in fact, for modeling ontologies intended for data-intensive tasks such as ontology-based data access~\cite{DBLP:conf/lics/CaliGLMP10,DBLP:conf/rweb/MugnierT14}.

Early characterizations of axiomatizability and finite axiomatizability by full existential rules (i.e., existential rules without existentially quantified variables) have been established by Makowsky and Vardi in the 1980s~\cite{MaVa86}. Moreover, ten Cate and Kolaitis obtained analogous characterizations for source-to-target existential rules, which have been used to formalize data exchange between a source schema and a target schema~\cite{CaKo09}. More recently, Console, Kolaitis and Pieris characterized finite axiomatizability by arbitrary existential rules, as well as existential rules from central guarded-based subclasses, that is, linear, guarded and frontier-guarded existential rules~\cite{CoKP21}.

To the best of our knowledge, the above results are the only known results in the literature concerning characterizations of axiomatizability and finite axiomatizability by rule-based formalisms.
Having said that, let us stress that several preservation results can be found in the literature. When the target is a preservation result, one considers two formalisms $\coll{L}$ and $\coll{L'}$, where $\coll{L'}$ is typically a proper fragment of $\coll{L}$, and the goal is to obtain characterizations of the following form: a set $\Phi$ of sentences from $\coll{L}$ is logically equivalent to a set of sentences from $\coll{L'}$ iff the models of $\Phi$ enjoy certain model-theoretic properties. A prototypical example of such a result is the \L{os}-Tarski Theorem~\cite{model-theory-book}, which states that a set $\Phi$ of first-order sentences is equivalent to a set of universal first-order sentences iff the collection of models of $\Phi$ is closed under substructures.
Such results for rule-based formalisms have been obtained in the context of description logics by Lutz et al.~\cite{LuPW11} and of existential rules by Zhang et al.~\cite{ZhZJ20}.

\subsection{Disjunctive Existential Rules}

The extension of existential rules with the feature of disjunction in the right-hand side of the implication,  originally proposed in~\cite{DeTa03}, leads to the central rule-based formalism of disjunctive existential rules that are first-order sentences of the form
\[
\forall \bar x \forall \bar y \left(\phi(\bar x,\bar y)\ \ra\ \bigvee_{i=1}^{k} \exists \bar z_i\, \psi_i(\bar x_i,\bar z_i)\right).
\]
Disjunctive existential rules have also found numerous applications in several areas such as data exchange~\cite{FKPT08}, expressive database query languages~\cite{EiGM97}, and knowledge representation and reasoning~\cite{AFLM12}, to name a few.

Despite the characterizations concerning existential rules discussed above, similar characterizations for disjunctive existential rules have remained largely unexplored so far. The main objective of this work is to change this state of affairs by focusing on finite axiomatizability. Our results towards this end can be summarized as follows:
\begin{itemize}
	\item We establish that a collection of structures $\coll{C}$ is finitely axiomatizable by disjunctive existential rules with at most $n$ universally quantified variables, at most $m$ existentially quantified variables, and at most $\ell$ disjuncts iff $\coll{C}$ is critical, closed under repairable direct products, and diagrammatically $(n,m,\ell)$-compatible. Criticality is a standard property, which has been used in several works (see, e.g.,~\cite{CoKP21}), and states that, for each integer $\kappa > 0$, $\coll{C}$ contains a structure $I_\kappa$ with $\kappa$ domain elements, and each relation of $I_\kappa$ contains all the tuples that can be formed using those $\kappa$ elements. Closure under repairable direct products is a new property obtained by carefully refining the standard property of closure under direct products since it is easy to show that disjunctive existential rules violate closure under direct products. Finally, diagrammatic $(n,m,\ell)$-compatibility, which is actually the main innovation of our characterization, relies on the method of diagrams.
	
	\item Diagrammatic $(n,m,\ell)$-compatibility turns out to be quite flexible. It can be tailored to other classes of disjunctive existential rules, so that it gives rise to the refined notions of linear-diagrammatic and guarded-diagrammatic $(n,m,\ell)$-compatibility. By exploiting these refined properties, we obtain characterizations of finite axiomatizability by linear and guarded disjunctive existential rules.
	
	\item Finally, we study the problem of rewriting a finite set of guarded disjunctive existential rules into an equivalent set that falls in the weaker class of linear disjunctive existential rules, whenever one exists. We provide an algorithm for this non-trivial problem, which heavily exploits the notion of linear-diagrammatic compatibility, that generalizes existing rewritings for non-disjunctive guarded existential rules presented in~\cite{CoKP21,ZhZJ20}.
\end{itemize}

%% file: preliminaries.tex
\section{Preliminaries}\label{sec:preliminaries}

Let $\ins{C}$ and $\ins{V}$ be disjoint countably infinite sets of constants and variables, respectively. For an integer $n > 0$, we may write $[n]$ for the set of integers $\{1,\ldots,n\}$.

\medskip

\noindent
\paragraph{Relational Structures.} 
A {\em (relational) schema} $\ins{S}$ is a finite set of relation symbols (or predicates) with positive arity; we write $\ar{R}$ for the arity of the relation symbol $R$.
A {\em (relational) structure} $I$ over a schema $\ins{S} = \{R_1,\ldots,R_n\}$, or {\em $\ins{S}$-structure}, is a tuple $(\adom{I},R_{1}^{I},\ldots,R_{n}^{I})$, where $\adom{I} \subseteq \ins{C}$ is a (finite or infinite) domain and $R_{1}^{I},\ldots,R_{n}^{I}$ are relations over $\adom{I}$, i.e., $R_{i}^{I} \subseteq \adom{I}^{\ar{R_i}}$ for each $i \in [n]$. We say that $I$ is a {\em finite structure} if $\adom{I}$ is finite.
A {\em fact} of $I$ is an expression of the form $R_i(\bar c)$, where $\bar c \in R_{i}^{I}$, and we denote by $\mathsf{facts}(I)$ the set of facts of $I$.
The {\em active domain} of a structure $I$, denoted $\aadom{I}$, is the set of elements of $\adom{I}$ that occur in at least one fact of $I$.
For an $\ins{S}$-structure $J = (\adom{J},R_{1}^{J},\ldots,R_{n}^{J})$, we write $J \subseteq I$ if $\facts{J} \subseteq \facts{I}$.
We say that $J$ is a {\em substrucure} of $I$, denoted $J \preceq I$, if $\adom{J} \subseteq \adom{I}$ and $R^{J} = R_{\mid \adom{J}}^{I}$ for each $R \in \ins{S}$ with $R_{\mid \adom{J}}^{I}$ being the {\em restriction of $R^{I}$ over $\adom{J}$}, i.e., the relation $\left\{\bar c \in R^{I} \mid \bar c \in \adom{J}^{\ar{R}}\right\}$.
Note that $J \preceq I$ implies $J \subseteq I$, but the other direction does not necessarily hold.
A {\em homomorphism} from $I$ to $J$ is a function $h : \adom{I} \ra \adom{J}$ such that, for each $i \in [n]$, $\bar c = (c_1,\ldots,c_m) \in R_{i}^{I}$ implies $h(\bar c) = (h(c_1),\ldots,h(c_m)) \in R_{i}^{J}$. We write $h : I \ra J$ for the fact that $h$ is a homomorphism from $I$ to $J$. Let $h(\facts{I})$ be the set $\{R(h(\bar c)) \mid R(\bar c) \in \facts{I}\}$.
Finally, we say that $I$ and $J$ are {\em isomorphic}, written $I \simeq J$, if there is a bijective homomorphism from $I$ to $J$ such that $h^{-1} : J \ra I$.

\medskip

\noindent
\paragraph{Extended Structures.} An {\em extended schema} is a pair $(\ins{S},\ins{F})$, where $\ins{S} = \{R_1,\ldots,R_n\}$ is a relational schema and $\ins{F} = \{f_1,\ldots,f_m\}$ is a set of $0$-ary function symbols.
A {\em structure} $I$ over $(\ins{S},\ins{F})$ is a tuple $(\adom{I},R_{1}^{I},\ldots,R_{n}^{I},f_{1}^{I},\ldots,f_{m}^{I})$, where $\adom{I} \subseteq \ins{C}$, $R_{1}^{I},\ldots,R_{n}^{I}$ are relations over $\adom{I}$, and $f_{1}^{I},\ldots,f_{m}^{I}$ are constants of $\adom{I}$.
Given a relational schema $\ins{S}$ and a finite non-empty set $C = \{c_1,\ldots,c_m\}$ of constants from $\ins{C}$, we denote by $\ins{S}[C]$ the extended schema $(\ins{S},\{f_{c_1},\ldots,f_{c_m}\})$. Furthermore, given an $\ins{S}$-structure $I = (\adom{I},R_{1}^{I},\ldots,R_{n}^{I})$, we denote by $I[C]$ the structure over $\ins{S}[C]$ defined as $(\adom{I} \cup C,R_{1}^{I},\ldots,R_{n}^{I},f_{c_1}^{I},\ldots,f_{c_m}^{I})$, where $f_{c_j}^{I} = c_j$ for each $j \in [m]$.\footnote{The additional notions defined above for relational structures (such as subset, substructure, homomorphism, and isomorphism) are not needed in the paper for structures over extended schemas.}

By abuse of notation, we may write down first-order formulas that mention a constant $c \in \ins{C}$, which is actually the $0$-ary function symbol $f_c$. For example, $\exists x (R(c,x,d) \wedge \neg (x = d))$, where $c,d \in \ins{C}$ and $x \in \ins{V}$, is essentially the sentence $\exists x (R(f_c,x,f_d) \wedge \neg (x = f_d))$ over the extended schema $(\{R\},\{f_c,f_d\})$.
Now, by abuse of terminology, given a structure $I$ over a schema $\ins{S}$ and a first-order sentence $\Phi$ that mentions relation symbols of $\ins{S}$ and constants (not necessarily from $\adom{I}$) of $\ins{C}$, we may say that $I$ satisfies $\Phi$, denoted $I \models \Phi$. With $C \subseteq \ins{C}$ being the set of constants occurring in $\Phi$, $I \models \Phi$ essentially denotes the fact that $I[C] \models \Phi$, that is, the structure $I[C]$ over the extended schema $\ins{S}[C]$ satisfies $\Phi$ under the standard first-order semantics.

\medskip

\noindent
\paragraph{Disjunctive Existential Rules.} An atom over a schema $\ins{S}$ is an expression of the form $R(\bar v)$, where $R \in \ins{S}$ and $\bar v$ is an $\ar{R}$-tuple of variables from $\ins{V}$.
A {\em disjunctive existential rule} (dexr) $\sigma$ over a schema $\ins{S}$ is a constant-free sentence
\[
\forall \bar x \forall \bar y \left(\phi(\bar x,\bar y)\ \ra\ \bigvee_{i=1}^{k} \exists \bar z_i\, \psi_i(\bar x_i,\bar z_i)\right),
\]
where $k > 0$, $\bar x, \bar y, \bar x_1,\ldots,\bar x_k,\bar z_1,\ldots,\bar z_k$ are tuples of variables of $\ins{V}$, the variables of $\bar x_i$ occur in $\bar x$ for $i \in [k]$, each variable of $\bar x$ occurs in $\bar x_i$ for some $i \in [k]$, $\phi(\bar x,\bar y)$ is a (possibly empty) conjunction of atoms over $\ins{S}$, and $\psi_i(\bar x_i,\bar z_i)$ is a non-empty conjunction of atoms over $\ins{S}$ for each $i \in [k]$.
For brevity, we write $\sigma$ as $\phi(\bar x,\bar y) \ra \bigvee_{i=1}^{k} \exists \bar z_i\, \psi_i(\bar x_i,\bar z_i)$ and use comma instead of $\wedge$ for joining atoms. 
When $\phi(\bar x,\bar y)$ is empty, $\sigma$ is essentially the sentence $\bigvee_{i=1}^{k} \exists \bar z_i\, \psi_i(\bar z_i)$.
We refer to $\phi(\bar x,\bar y)$ and $\bigvee_{i=1}^{k} \psi_i(\bar x_i,\bar z_i)$ as the {\em body} and {\em head} of $\sigma$, denoted $\body{\sigma}$ and $\head{\sigma}$, respectively.
By abuse of notation, we may treat a tuple of variables as a set of variables and a conjunction of atoms as a set of atoms.
Assuming that $\phi(\bar x,\bar y)$ is non-empty, an $\ins{S}$-structure $I$ satisfies $\sigma$ if the following holds: whenever there exists a function $h : \bar x \cup \bar y \ra \adom{I}$ such that $h(\phi(\bar x, \bar y)) \subseteq \facts{I}$ (as usual, we write $h(\phi(\bar x,\bar y))$ for the set $\{R(h(\bar v)) \mid R(\bar v) \in \phi(\bar x,\bar y)\}$), then there exists an integer $i \in [k]$ and an extension $h'$ of $h$ such that $h'(\psi_i(\bar x_i, \bar z_i)) \subseteq \facts{I}$. 
When $\phi(\bar x,\bar y)$ is empty, an $\ins{S}$-structure $I$ satisfies $\sigma$ if there exists $i \in [k]$ and a function $h : \bar z_i \ra \adom{I}$ such that $h(\psi_i(\bar z_i)) \subseteq \facts{I}$.
We write $I \models \sigma$ for the fact that $I$ satisfies $\sigma$. The $\ins{S}$-structure $I$ satisfies a set $\dep$ of dexrs over $\ins{S}$, written $I \models \dep$, in which case we say that $I$ is a {\em model} of $\dep$, if $I \models \sigma$ for each $\sigma \in \dep$

\medskip

\noindent
\paragraph{Finite Axiomatizability.} Let $\coll{C}$ be a collection of structures over a schema $\ins{S}$. We say that $\coll{C}$ is {\em finitely axiomatizable} by dexrs if there exists a finite set $\dep$ of dexrs over $\ins{S}$ such that $I \in \coll{C}$ iff $I \models \dep$, that is, $\coll{C}$ is precisely the set of models of $\dep$.
Henceforth, since dexrs cannot distinguish isomorphic structures, we implicitly assume that collections $\coll{C}$ of structures are {\em closed under isomorphisms}, i.e., if $I \in \coll{C}$ and $J$ is a structure such that $I \simeq J$, then $J \in \coll{C}$. 

%% file: properties.tex
\section{Model-theoretic Properties}\label{sec:properties}

We now introduce three model-theoretic properties of collections of structures that will play a crucial role in our characterizations. Fix an arbitrary schema $\ins{S} = \{R_1,\ldots,R_r\}$.

\subsection{Criticality}

An $\ins{S}$-structure $I = (\adom{I},R_{1}^{I},\ldots,R_{r}^{I})$ is {\em $\kappa$-critical}, for $\kappa>0$, if $|\adom{I}| = \kappa$ and $R_{i}^{I} = \adom{I}^{\ar{R_i}}$ for each $i \in [r]$.
We can now lift criticality to collections of structures.

\begin{definition}\label{def:criticality}
	A collection $\coll{C}$ of $\ins{S}$-structures is $\kappa$-critical, for $\kappa > 0$, if it contains a $\kappa$-critical $\ins{S}$-structure. We further say that $\coll{C}$ is {\em critical} if it is $\kappa$-critical for each $\kappa>0$. \hfill\markfull
\end{definition}

It is easy to show the following lemma:

\begin{lemma}\label{lem:dexrs-critical}
	A collection $\coll{C}$ of structures that is finitely axiomatizable by dexrs is critical.
\end{lemma}

\subsection{Closure Under Repairable Direct Products}

Consider two $\ins{S}$-structures $I = (\adom{I},R_{1}^{I},\ldots,R_{r}^{I})$ and $J = (\adom{J},R_{1}^{J},\ldots,R_{r}^{J})$.
Their {\em direct product}, denoted $I \otimes J$, is the $\ins{S}$-structure $K = (\adom{K},R_{1}^{K},\ldots,R_{r}^{K})$, where $\adom{K} = \adom{I} \times \adom{J}$ and
\begin{multline*}
R_{i}^{K}\ =\ \left\{\left((a_1,b_1),\ldots,(a_{\ar{R_i}},b_{\ar{R_i}})\right) \mid\right.\\ \left.\left(a_1,\ldots,a_{\ar{R_i}}\right) \in R_{i}^{I} \text{ and } \left(b_1,\ldots,b_{\ar{R_i}}\right) \in R_{i}^{J}\right\},
\end{multline*}
for each $i \in [r]$.
A collection of structures $\coll{C}$ is closed under direct products if, for every $I,J \in \coll{C}$, the direct product of $I$ and $J$ belongs to $\coll{C}$. 
Closure under direct products is a well-known property that has been extensively used in model theory. However, it is not appropriate towards a characterization of finite axiomatizability by dexrs.

\begin{example}\label{exa:direct-product}
	Consider the set $\dep$ consisting of the single dexr
	\[
	R(x)\ \ra\ S(x) \vee T(x)
	\]
	and the structures $I_1$ and $I_2$ with $\facts{I_1} = \{R(a),S(a)\}$ and $\facts{I_2} = \{R(a),T(a)\}$, and thus,
	$
	\facts{I_1 \otimes I_2} = \{R(aa)\}.
	$
	Clearly, $I_1 \models \dep$ and $I_2 \models \dep$, but $I_1 \otimes I_2 \not\models \dep$. \hfill\markfull
\end{example}

The above example essentially tells us that a collection of structures that is finitely axiomatizable by dexrs is not necessarily closed under direct products, and thus, we need to revisit the standard notion of direct product.
To this end, we define the function $\pi_K : \adom{K} \ra \adom{I}$ as follows: for every $(a_I,a_J) \in \adom{K}$, $\pi_K((a_I,a_J)) = a_I$. It is clear that $\pi_K$ is a homomorphism from $K$ to $I$, which we call the {\em projective homomorphism for K}.
Now, a {\em repairable direct product} of the $\ins{S}$-structures $I$ and $J$ is an $\ins{S}$-structure $L$ such that 
	(1) $I \otimes J \subseteq L$, and 
	(2) there exists an extension $h$ of $\pi_K$ that is a homomorphism from $L$ to $I$.
We now define the following novel model-theoretic property:

\begin{definition}\label{def:closure-under-products}
	A collection $\coll{C}$ of $\ins{S}$-structures is {\em closed under repairable direct products} if, for every two $\ins{S}$-structures $I,J \in  \coll{C}$, there exists an $\ins{S}$-structure in $\coll{C}$ that is a repairable direct product of $I$ and $J$. \hfill\markfull
\end{definition}

Coming back to Example~\ref{exa:direct-product}, it is easy to verify that the structure $L$ with $\facts{L} = \{R(aa),S(aa)\}$ is a repairable direct product of $I_1$ and $I_2$; indeed, $I_1 \otimes I_2 \subseteq L$ and $\pi_{I_1 \otimes I_2}$ is a homomorphism from $L$ to $I_1$. We can show the following:

\begin{lemma}\label{lem:dexrs-closed-under-rpds}
	A collection $\coll{C}$ of structures that is finitely axiomatizable by dexrs is closed under repairable direct products.
\end{lemma}

\subsection{Diagrammatic Compatibility}\label{sec:locality}

We now introduce our new property of collections of structures, which in turn relies on the notion of relative diagram of a finite structure. We first introduce relative diagrams and then define the new model-theoretic property of interest.

\medskip

\noindent\paragraph{Relative Diagrams.}
Consider a structure $I$ over a schema $\ins{S}$ and a finite structure $K \subseteq I$ with $\adom{K} = \aadom{K}$. We are interested in the so-called $(m,\ell)$-diagrams of $K$ relative to $I$ for integers $m,\ell \geq 0$, which we define below.
To this end, let $A_{K,m}$ be the set of all atomic formulas of the form $R(\bar u)$ that can be formed using predicates from $\ins{S}$, constants from $\adom{K}$, and $m$ distinct variables $y_1,\ldots,y_m$ from $\ins{V}$, i.e., $R \in \ins{S}$ and $\bar u \in (\adom{K} \cup \{y_1,\ldots,y_m\})^{\ar{R}}$.
Furthermore, let $C_{K,m}$ be the set of all conjunctions of atomic formulas from $A_{K,m}$. Observe that both $A_{K,m}$ and $C_{K,m}$ are finite sets since $\adom{K}$ is finite. Let
\[
N_{K,m}^{I}\ =\ \{\gamma(\bar y) \in C_{K,m} \mid I \not\models \exists \bar y \, \gamma(\bar y)\},
\]
which is clearly finite since $C_{K,m}$ is finite.
For a set of formulas $G \subseteq N_{K,m}^{I}$, the {\em $G$-diagram of $K$ relative to $I$}, denoted $\Delta_{K,G}^{I}$, is defined as the first-order sentence
\[
\underbrace{\bigwedge_{\alpha \in \facts{K}} \alpha}_{\Psi_1}\ \wedge\ \underbrace{\bigwedge_{\substack{c,d \in \adom{K}, \\ c \neq d}} \neg (c=d)\ \wedge\ \\ \bigwedge_{\gamma(\bar y) \in G} \neg \left(\exists \bar y \, \gamma(\bar y)\right)}_{\Psi_2}.
\]
Note that $\Psi_1$ and $\Psi_2$ might be empty (i.e., they have no conjuncts). In particular, $\Psi_1$ is empty if $K$ is empty, whereas $\Psi_2$ is empty if $I$ is $1$-critical, which means that $|\adom{K}|=1$ and $I \models \exists \bar y \, \gamma(\bar y)$ for each $\gamma(\bar y) \in C_{K,m}$, and thus, $G = \emptyset$. When $\Psi_1$ and $\Psi_2$ are empty, $\Delta_{K,G}^{I}$ is the truth value $\mathsf{true}$, i.e., a tautology.
Intuitively, $\Delta_{K,G}^{I}$ provides a witness for the fact that a dexr, whose head-disjuncts belong to $G$, is violated by the structure $I$ due to the structure $K \subseteq I$.

\begin{definition}
Consider an $\ins{S}$-structure $I$ and a finite structure $K \subseteq I$ with $\adom{K} = \aadom{K}$. 
An {\em $(m,\ell)$-diagram of $K$ relative to $I$}, for $m,\ell \geq 0$, is a $G$-diagram of $K$ relative to $I$, where $G \subseteq N_{K,m}^{I}$ and $|G| \leq \ell$. \hfill\markfull
\end{definition}

Let $\Phi_{K,G}^{I}(\bar x)$ be the formula obtained from $\Delta_{K,G}^{I}$ by replacing each constant $c \in \adom{K}$ with a new variable $x_c \in \ins{V} \setminus \{y_1,\ldots,y_m\}$. 
If $\Delta_{K,G}^{I}$ is empty, then $\Phi_{K,G}^{I}$ is the truth value $\mathsf{true}$, i.e., a tautology.
It is easy to verify the following lemma, which will be used in our proofs:

\begin{lemma}\label{lem:sat-diagram}
	It holds that $I \models \exists \bar x \, \Phi_{K,G}^{I}(\bar x)$.
\end{lemma}

\noindent\paragraph{The Property.} Consider a collection $\coll{C}$ of structures over a schema $\ins{S}$ and an $\ins{S}$-structure $I$. For $n,m,\ell\geq 0$, we say that $\coll{C}$ is {\em diagrammatically $(n,m,\ell)$-compatible} with $I$ if, for every $K \preceq I$ with $\adom{K} = \aadom{K}$ and $|\adom{K}| \leq n$, and every $(m,\ell)$-diagram $\Delta_{K,G}^{I}$ of $K$ relative to $I$, where $G \subseteq N_{K,m}^{I}$, there is $J \in \coll{C}$ such that $J \models \Delta_{K,G}^{I}$.
Roughly, this states that no matter how a dexr $\sigma$ with a bounded number of variables (at most $n$ universally quantified and at most $m$ existentially quantified variables) and a bounded number of disjuncts in its head (at most $\ell$ disjuncts) is violated by the structure $I$, there is always a structure $J \in \coll{C}$ that violates $\sigma$ for the same reason.
The new property of interest follows.

\begin{definition}\label{def:diagrammatic-compatibility}
	A collection $\coll{C}$ of $\ins{S}$-structures is {\em diagrammatically $(n,m,\ell)$-compatible}, for $n,m,\ell \geq 0$, if, for every $\ins{S}$-structure $I$, the following holds: $\coll{C}$ is diagrammatically $(n,m,\ell)$-compatible with $I$ implies $I \in \coll{C}$. \hfill\markfull
\end{definition}

The next lemma establishes that collections of structures that are finitely axiomatizable by dexrs are diagrammatically compatible. Actually, it shows a stronger claim since it relates the integers $n,m$ and $\ell$ that witness diagrammatic $(n,m,\ell)$-compatibility with the number of universally quantified variables, existentially quantified variables, and head-disjuncts, respectively, occurring in the dexrs.
A dexr is called {\em $(n,m,\ell)$-dexr}, for $n,m \geq 0$, with $n+m > 0$, and $\ell > 0$, if it mentions at most $n$ universally quantified variables in its body, at most $m$ existentially quantified variables in each disjunct of its head, and at most $\ell$ disjuncts in its head.
Note that we require $n+m > 0$ since, by definition, a dexr has at least one variable that is either universally or existentially quantified. Furthermore, we require $\ell > 0$ since, by definition, a dexr has at least one disjunct in its head.

\begin{lemma}\label{lem:dexr-locality}
		 For $n,m \geq 0$, with $n+m > 0$, and $\ell > 0$, every collection $\coll{C}$ of structures that is finitely axiomatizable by $(n,m,\ell)$-dexrs is diagrammatically $(n,m,\ell)$-compatible.
\end{lemma}

\begin{proof}
	Let $\dep$ be a finite set of $(n,m,\ell)$-dexrs with $I \in \coll{C}$ iff $I \models \dep$. Consider a structure $I$ and assume that $\coll{C}$ is diagrammatically $(n,m,\ell)$-compatible with $I$. We proceed to show that $I \in \coll{C}$, or, equivalently, $I \models \dep$.
	Consider a dexr $\sigma \in \dep$ of the form $\phi(\bar x,\bar y) \ra \bigvee_{i=1}^{k} \exists \bar z_i\, \psi_i(\bar x_i,\bar z_i)$. We assume that $\phi(\bar x,\bar y)$ is non-empty; the case where $\phi(\bar x,\bar y)$ is empty is treated similarly. 
	Assume that there exists a function $h : \bar x \cup \bar y \ra \adom{I}$ such that $h(\phi(\bar x, \bar y)) \subseteq \facts{I}$. We need to show that there exists $i \in [k]$ such that $\lambda(\psi_i(\bar x_i, \bar z_i)) \subseteq \facts{I}$ with $\lambda$ being an extension of $h$.
	By contradiction, assume that such $i \in [k]$ does not exist. This means that, for every $i \in [k]$, $I \models \neg \exists \bar z_i \, \psi_i'(\bar z_i)$, where $\psi_i'(\bar z_i)$ is obtained from $\psi_i(\bar x_i,\bar z_i)$ by replacing each variable $x$ of $\bar x_i$ with $h(x)$.
	Let $K$ be the structure $(\adom{K},R_{1}^{K},\ldots,R_{r}^{K})$, where $\adom{K}$ is the set of constants occurring in $h(\phi(\bar x,\bar y))$, and, for each $i \in [r]$, $R_{i}^{K} = {R_{i}^{I}}_{|K}$. It is clear that $K \preceq I$ with $\adom{K} = \aadom{K}$ and $|\aadom{K}| \leq n$ since $\phi(\bar x,\bar y)$ mentions at most $n$ variables. 
	Let $G = \{\psi'_i(\bar z_i)\}_{i \in [k]}$. Since, for each $i \in [k]$, $\psi_i(\bar x_i,\bar z_i)$ mentions at most $m$ existentially quantified variables and $k \leq \ell$, it is easy to see that $G \subseteq N_{K,m}^{I}$ and $|G| \leq \ell$. Therefore, $\Delta_{K,G}^{I}$, that is, the $G$-diagram of $K$ relative to $I$ is an $(m,\ell)$-diagram of $K$ relative to $I$. Since $\coll{C}$ is diagrammatically $(n,m,\ell)$-compatible with $I$, we get that there exists $J \in \coll{C}$ that satisfies $\Delta_{K,G}^{I}$. The latter implies that $h(\phi(\bar x,\bar y)) \subseteq \facts{J}$, but there is no extension $\lambda$ of $h$ such that $\lambda(\psi_i(\bar x_i,\bar z_i)) \subseteq \facts{J}$ for some $i \in [k]$. Consequently, $J \not\models \sigma$, and thus, $J \not\models \dep$. But this contradicts the fact that $J \in \coll{C}$, which is equivalent to say that $J \models \dep$, and the claim follows.
\end{proof}

%% file: ders.tex
\section{Finite Axiomatizability by Disjunctive Existential Rules}\label{sec:dexrs}

We proceed to characterize when a collection $\coll{C}$ of structures is finitely axiomatizable by dexrs. More precisely, the goal of this section is to establish the following result:

\begin{theorem}\label{the:dexr-refined-characterization}
	For a collection $\coll{C}$ of structures and $n,m \geq 0$, with $n+m > 0$, and $\ell > 0 $, the following are equivalent:
	\begin{enumerate}
		\item $\coll{C}$ is finitely axiomatizable by $(n,m,\ell)$-dexrs.
		\item $\coll{C}$ is critical, closed under repairable direct products, and diagrammatically $(n,m,\ell)$-compatible.
	\end{enumerate}
\end{theorem}

It is clear that the direction $(1) \Rightarrow (2)$ immediately follows from Lemmas~\ref{lem:dexrs-critical},~\ref{lem:dexrs-closed-under-rpds} and~\ref{lem:dexr-locality}. The rest of this section is devoted to discussing the proof of $(2) \Rightarrow (1)$. To this end, we need to introduce disjunctive dependencies.

\medskip
\noindent \paragraph{Disjunctive Dependencies.} A {\em disjunctive dependency} (dd) $\delta$ over a schema $\ins{S}$ is a constant-free sentence of the form
\[
\forall \bar x \, \left(\phi(\bar x)\ \ra\ \bigvee_{i=1}^{k}
\psi_i(\bar x_i)\right),
\]
where $k \geq 0$, $\bar x$ is a (possibly empty) tuple of variables of $\ins{V}$, the expression $\phi(\bar x)$ is a (possibly empty) conjunction of atoms over $\ins{S}$, and, assuming $k>0$, for each $i \in [k]$, $\bar x_i \subseteq \bar x$ and the expression $\psi_i(\bar x_i)$ is either an equality formula $y=z$ with $\bar x_i = \{y,z\}$, or a formula $\exists \bar y_i \chi_i(\bar x_i,\bar y_i)$ with $\bar y_i$ being a tuple of variables from $\ins{V} \setminus \bar x$ and $\chi_i(\bar x_i,\bar y_i)$ a (non-empty) conjunction of atoms over $\ins{S}$.
When $k=0$, there are no disjuncts in the conclusion of $\delta$. In this case, if $\phi(\bar x)$ is empty, then $\delta$ is interpreted as the truth value $\mathsf{false}$, i.e., a contradiction; otherwise, $\delta$ is essentially the sentence $\forall \bar x \left(\phi(\bar x) \ra \mathsf{false}\right) \equiv \neg(\exists \bar x \, \phi(\bar x))$.
Now, when $k > 0$ and $\phi(\bar x)$ is empty, $\delta$ is essentially the sentence $\bigvee_{i=1}^{k} \psi_i$. 
If $k > 0$ and, for each $i \in [k]$, $\psi_i(\bar x_i)$ is an equality formula, then $\delta$ is called a {\em disjunctive equality rule} (deqr).

Assuming that $k = 0$ and $\phi(\bar x)$ is non-empty, an $\ins{S}$-structure $I$ satisfies $\delta$ if there is no function $h : \bar x \ra \adom{I}$ such that $h(\phi(\bar x)) \subseteq \facts{I}$.
Assume now that $k > 0$. If $\phi(\bar x)$ is non-empty, then $\delta$ is satisfied by an $\ins{S}$-structure $I$ if, whenever there exists a function $h : \bar x \ra \adom{I}$ such that $h(\phi(\bar x)) \subseteq \facts{I}$, then there is $i \in [k]$ such that, if $\psi_i(\bar y_i)$ is $y=z$, then $h(y) = h(z)$; otherwise, if $\psi_i(\bar x_i)$ is $\exists \bar y_i \, \chi_i(\bar x_i,\bar y_i)$, then there is an extension $h'$ of $h$ such that $h'(\chi_i(\bar x_i,\bar y_i)) \subseteq \facts{I}$. 
In case $\phi(\bar x)$ is empty, then $\delta = \bigvee_{i \in [k]} \exists \bar y_i \, \chi_i(\bar y_i)$ is satisfied by $I$ if there is $i \in [k]$ and a function $h : \bar y_i \ra \adom{I}$ such that $h(\chi_i(\bar y_i)) \subseteq \facts{I}$.
We write $I \models \delta$ for the fact that $I$ satisfies $\delta$. The structure $I$ satisfies a set $\dep$ of dds, written $I \models \dep$, in which case we say that $I$ is a {\em model} of $\dep$, if $I \models \delta$ for each $\delta \in \dep$.

For a schema $\ins{S}$, let $\class{DD}_{n,m,\ell}^{\ins{S}}$, for integers $n,m,\ell \geq 0$, be the set of dds over $\ins{S}$ of the form
\[
\forall \bar x \left(\phi(\bar x)\ \ra\ \bigvee_{i=1}^{k} \psi_i(\bar x_i)\right),
\]
where $k \geq 0$, such that 
\begin{enumerate}
	\item $\bar x$ consists of at most $n$ distinct variables, 
	\item if $k \geq 1$, then, for each $i \in [k]$, if $\psi_i(\bar x_i)$ is a formula of the form $\exists \bar y_i \, \chi_i(\bar x_i,\bar y_i)$ with $\chi_i(\bar x_i,\bar y_i)$ being a non-empty conjunction of atoms, then $\bar y_i$ consists of at most $m$ distinct variables, and 
	\item if $k \geq 1$, then it holds that
	\[
	|\{i \in [k] \mid \psi_i(\bar x_i) \text{ is not an equality formula}\}| \leq \ell,
	\]
	i.e., at most $\ell$ disjuncts are a non-empty conjunction of atoms (i.e., not an equality formula).
\end{enumerate}
Note that $\class{DD}_{n,m,\ell}^{\ins{S}}$ is a finite set (up to variable renaming) since $\ins{S}$ is finite, and the number of variables and number of disjuncts in each element of $\class{DD}_{n,m,\ell}^{\ins{S}}$ is finite.

\subsection{Proving the Direction $(2) \Rightarrow (1)$}

We now have all the ingredients needed for discussing the proof of the direction $(2) \Rightarrow (1)$ of Theorem~\ref{the:dexr-refined-characterization}.
Consider a collection $\coll{C}$ of structures over a schema $\ins{S}$ that is critical, closed under repairable direct products, and diagrammatically $(n,m,\ell)$-compatible for integers $n,m \geq 0$, with $n+m >0$, and $\ell > 0$. We proceed to show that $\coll{C}$ is finitely axiomatizable by $(n,m,\ell)$-dexrs in three steps:
\begin{enumerate}
	\item We first define a finite set $\dep^\vee$ of dds from $\class{DD}_{n,m,\ell}^{\ins{S}}$ such that, for every $\ins{S}$-structure $I$, $I \in \coll{C}$ iff $I \models \dep^\vee$. This exploits the fact that $\coll{C}$ is $1$-critical (since it is critical) and diagrammatically $(n,m,\ell)$-compatible.
	
	\item We then show that there is a finite set $\dep^{\exists,=}$ of $(n,m,\ell)$-dexrs and $n$-deqrs (i.e., deqrs with at most $n$ universally quantified variables with the corner case of $0$-deqr being the value $\mathsf{true}$) over $\ins{S}$ such that $\dep^\vee \equiv \dep^{\exists,=}$; in fact, $\dep^{\exists,=}$ is the set of dexrs and deqrs occurring in $\dep^{\vee}$. This exploits the fact that $\coll{C}$ is closed under repairable direct products.

	\item We finally argue that $\dep^{\exists,=}$ consists only of dexrs, which implies that $\coll{C}$ is finitely axiomatizable by $(n,m,\ell)$-dexrs. This exploits the fact that $\coll{C}$ is critical.
\end{enumerate}
We proceed to give further details for the above three steps.

\medskip

\noindent
\paragraph{\underline{Step 1: The finite set $\dep^{\vee}$ of dds}}

\smallskip

\noindent  
Let $\dep^{\vee}$ be the set of all dds from $\class{DD}_{n,m,\ell}^{\ins{S}}$ that are satisfied by every structure of $\coll{C}$, that is,
\[
\dep^{\vee}\ =\ \left\{\delta \in \class{DD}_{n,m,\ell}^{\ins{S}} \mid \text{ for each } I \in \coll{C}, \text{ we have } I \models \delta\right\}.
\]
Clearly, the set $\dep^{\vee}$ is finite (up to variable renaming) since $\dep^{\vee} \subseteq \class{DD}_{n,m,\ell}^{\ins{S}}$. We proceed to show that $\coll{C}$ is precisely the set of $\ins{S}$-structures that satisfy $\dep^{\vee}$.

\begin{lemma}\label{lem:tgd-egd-dc-characterization-lemma-1}
	For every $\ins{S}$-structure $I$, $I \in \coll{C}$ iff $I \models \dep^{\vee}$.
\end{lemma}

\begin{proof}
	The $(\Rightarrow)$ direction holds by construction. We proceed with the non-trivial direction $(\Leftarrow)$. Consider an $\ins{S}$-structure $I$ such that $I \models \dep^\vee$. We are going to show that $\coll{C}$ is diagrammatically $(n,m,\ell)$-compatible with $I$, which in turn implies that $I \in \coll{C}$ since $\coll{C}$ is diagrammatically $(n,m,\ell)$-compatible. Consider an arbitrary substructure $K$ of $I$ with $\adom{K} = \aadom{K}$ and $|\adom{K}| \leq n$, and an arbitrary $(m,\ell)$-diagram $\Delta_{K,G}^{I}$ of $K$ relative to $I$, where $G \subseteq N_{K,m}^{I}$. We need to show that there exists $J \in \coll{C}$ that satisfies $\Delta_{K,G}^{I}$. To this end, we first establish the following auxiliary claim:
	
	\begin{claim}\label{cla:equivalent-dd}
		There is $\delta \in \class{DD}_{n,m,\ell}^{\ins{S}}$ with $\delta \equiv \neg \exists \bar x \, \Phi_{K,G}^{I}(\bar x)$.
	\end{claim}

Let $\delta \in \class{DD}_{n,m,\ell}^{\ins{S}}$ be the dd provided by Claim~\ref{cla:equivalent-dd} such that $\delta \equiv \neg \exists \bar x \, \Phi_{K,G}^{I}(\bar x)$. 
We claim that $\delta \not\in \dep^\vee$. By contradiction, assume that $\delta \in \dep^\vee$. This implies that $I \models \delta$, which cannot be the case since, by Lemma~\ref{lem:sat-diagram}, $I \models \exists \bar x \, \Phi_{K,G}^{I}(\bar x)$. The fact that $\delta \not\in \dep^\vee$ implies that there exists an $\ins{S}$-structure $L \in \coll{C}$ such that $L \not\models \delta$, which means that $L \models \exists \bar x \, \Phi_{K,G}^{I}(\bar x)$. Therefore, there is an $\ins{S}$-structure $J$ such that $J \simeq L$ and $J \models \Delta_{K,G}^{I}$. Since $\coll{C}$ is closed under isomorphisms, we can conclude that $J \in \coll{C}$, and the claim follows.
\end{proof}

\noindent
\paragraph{\underline{Step 2: The finite set $\dep^{\exists,=}$ of dexrs and deqrs}}

\smallskip

\noindent It is clear that $\dep^\vee \models \dep^{\exists,=}$, that is, each model of $\dep^\vee$ is a model of $\dep^{\exists,=}$, since $\dep^{\exists,=} \subseteq \dep^\vee$. It remains to show that $\dep^{\exists,=} \models \dep^\vee$. By contradiction, assume that $\dep^{\exists,=} \not\models \dep^\vee$. This implies that there exists a dd $\delta \in \dep^\vee$ such that $\dep^{\exists,=} \not\models \delta$. Clearly, $\delta$ is neither a dexr nor a deqr. Thus, $\delta$ can be written as a sentence of the form
\[
\forall \bar x \, \left(\phi(\bar x)\ \ra\ \bigvee_{i=1}^{k_1}
(z_i = w_i)\ \vee\ \bigvee_{i=1}^{k_2}
\exists \bar y_i \, \chi_i(\bar x_i,\bar y_i)\right),
\]
where $k_1,k_2 \geq 1$. Let $\delta_=$ be the deqr
\[
\forall \bar x \, \left(\phi(\bar x)\ \ra\ \bigvee_{i=1}^{k_1}
(z_i = w_i)\right).
\]
Since $\dep^{\exists,=} \not\models \delta$, we can conclude that $\dep^\vee \not\models \delta_=$; otherwise, $\delta_= \in \dep^{\exists,=}$ which cannot be the case. Therefore, we get that there exists an $\ins{S}$-structure $I_= \in \coll{C}$ such that $I_= \not\models \delta_=$. Analogously, with $\delta_\exists$ being the dexr
\[
\forall \bar x \, \left(\phi(\bar x)\ \ra\ \bigvee_{i=1}^{k_2}
\exists \bar y_i \, \chi_i(\bar x_i,\bar y_i)\right),
\]
we can show that there exists an $\ins{S}$-structure $I_\exists \in \coll{C}$ such that $I_\exists \not\models \delta_\exists$. We now proceed, by exploiting the structures $I_=$ and $I_\exists$, to show that there is an $\ins{S}$-structure $J$ such that $J \in \coll{C}$ and $J \not\models \delta$, which leads to a contradiction.
Since $\coll{C}$ is closed under repairable direct products, we get that there exists an $\ins{S}$-structure $J \in \coll{C}$ that is a repairable direct product of $I_\exists$ and $I_=$, which means that (i) $I_\exists \otimes I_= \subseteq J$, and (ii) there exists an extension $h_{I_\exists \otimes I_=}$ of $\pi_{I_\exists \otimes I_=}$ that is a homomorphism from $J$ to $I_\exists$. It remains to show that $J \not\models \delta$.

Since $I_= \not\models \delta_=$, we get that there exists a function $h_= : \bar x \ra \adom{I_=}$ such that $h_=(\phi(\bar x)) \subseteq \facts{I_=}$ and $h_=(z_i) \neq h_=(w_i)$, for each $i \in [k_1]$. Similarly, since $I_\exists \not\models \delta_\exists$, we get that there exists a function $h_\exists : \bar x \ra \adom{I_\exists}$ such that $h_\exists(\phi(\bar x)) \subseteq \facts{I_\exists}$ and there is no extension $h'_{\exists}$ of $h_\exists$ such that $h'_\exists(\chi_i(\bar x_i,\bar y_i)) \subseteq \facts{I_\exists}$, for each $i \in [k_2]$.
Let $L_{\phi}^{\exists}$ be the structure with $\adom{L_{\phi}^{\exists}} = h_\exists(\bar x)$ and $\facts{L_{\phi}^{\exists}} = h_\exists(\phi(\bar x))$. Analogously, let $L_{\phi}^{=}$ be the structure with $\adom{L_{\phi}^{=}} = h_=(\bar x)$ and $\facts{L_{\phi}^{=}} = h_=(\phi(\bar x))$. It is easy to verify that there exists a function $h : \bar x \ra \adom{L_{\phi}^{\exists} \otimes L_{\phi}^{=}}$ such that $h(\phi(\bar x)) \subseteq \facts{L_{\phi}^{\exists} \otimes L_{\phi}^{=}}$; in particular, for each $x \in \bar x$, $h(x) = (h_\exists(x),h_=(x))$. We show that $h$ witnesses the fact that $J \not\models \delta$.

Clearly, $L_{\phi}^{\exists} \otimes L_{\phi}^{=} \subseteq I_\exists \otimes I_= \subseteq J$, which implies that $h(\phi(\bar x)) \subseteq \facts{J}$. It remains to show that (i) for each $i \in [k_1]$, $h(z_i) \neq h(w_i)$, and (ii) for each $i \in [k_2]$, there is no extension $h'$ of $h$ such that $h'(\chi_i(\bar x_i,\bar y_i)) \subseteq \facts{J}$.
Concerning item (i), we observe that $h_=(z_i) \neq h_=(w_i)$ since $h_=$ witnesses the fact that $I_= \not\models \delta_=$. Therefore, by definition of $h$, we get that $h(z_i) \neq h(w_i)$, as needed.
Concerning item (2), by contradiction, assume that there exists $i \in [k_2]$ and an extension $h'$ of $h$ such that $h'(\chi_i(\bar x_i,\bar y_i)) \subseteq \facts{J}$. Let $h'' = h_{I_\exists \otimes I_=} \circ h'$. By definition of $h$, we get that $h''(x) = h_\exists(x)$ for each $x \in \bar x$. Moreover, by composition, we get that $h''(\chi_i(\bar x_,\bar y_i)) \subseteq \facts{I_\exists}$. The latter implies that $h''$ is an extension of $h_\exists$ such that $h''(\chi(\bar x_i,\bar y_i)) \subseteq \facts{I_\exists}$, which contradicts the fact that $h_\exists$ witnesses $I_\exists \not\models \delta_\exists$.

\medskip

\noindent
\paragraph{\underline{Step 3: The set $\dep^{\exists,=}$ consists only of dexrs}}

\smallskip

\noindent By contradiction, assume that a deqr $\delta$ of the form
\[
\forall \bar x \, \left(\phi(\bar x)\ \ra\ \bigvee_{i=1}^{k}
(z_i = w_i)\right),
\]
where $k \geq 1$, occurs in $\dep^{\exists,=}$. Since $\coll{C}$ is critical, there exists a $|\bar x|$-critical structure $I \in \coll{C}$, where $|\bar x|$ is the number of distinct variables in $\bar x$, and a function $h : \bar x \ra \adom{I}$ such that $h(\phi(\bar x)) \subseteq \facts{I}$ and $h(z_i) \neq h(w_i)$ for each $i \in [k]$. This implies that $I \not\models \delta$, which contradicts the fact that every dd of $\dep^{\exists,=}$ is satisfied by every structure of $\coll{C}$.

\subsection{Diagrammatic Compatibility vs Locality}

As discussed in Section~\ref{sec:introduction}, a characterization similar to Theorem~\ref{the:dexr-refined-characterization} for existential rules (i.e., the special case of disjunctive existential rules where the head consists of exactly one disjunct) has been established in~\cite{CoKP21} by exploiting a model-theoretic property called locality; for our discussion, the formal definition of locality is not crucial and is deferred to the appendix. More precisely, it was shown that, for a collection $\coll{C}$ of structures and $n,m \geq 0$, with $n+m>0$, the following are equivalent:
\begin{enumerate}
	\item $\coll{C}$ is finitely axiomatizable by existential rules with at most $n$ universally and $m$ existentially quantified variables.
	\item $\coll{C}$ is critical, closed under direct products, and $(n,m)$-local.
\end{enumerate}
The question that comes up is whether we could use the notion of locality, instead of introducing the new property of diagrammatic compatibility, for achieving the characterization for dexrs stated in Theorem~\ref{the:dexr-refined-characterization}. It turns out, as shown below, that this is not the case, which justifies the introduction of diagrammatic compatibility that can be understood as a refined notion of locality that explicitly takes into account the
number of disjuncts in the head of a rule.

\begin{proposition}\label{pro:compatibility-implies-locality}
	Consider a collection $\coll{C}$ of structures. For every $n,m \geq 0$, with $n+m>0$, and $\ell >0$, $\coll{C}$ is diagrammatically $(n,m,\ell)$-compatible implies $\coll{C}$ is $(n,m)$-local.
\end{proposition}

The above proposition implies that every collection of structures that is finitely axiomatizable by $(n,m,\ell)$-dexrs is $(n,m)$-local. This means that $(n,m)$-locality is not powerful enough to distinguish between classes of structures $\coll{C}$ and $\coll{C}'$ such that $\coll{C}$ is finitely axiomatizable by $(n,m,\ell)$-dexrs and $\coll{C}'$ is finitely axiomatizable by $(n,m,\ell')$-dexrs, for $\ell \neq \ell'$. To further illustrate this fact, let us consider again the set $\dep$ of dexrs from Example~\ref{exa:direct-product} consisting of the dexr $R(x) \ra\ S(x) \vee T(x)$ and let $\coll{C}_\dep$ be the set of models of $\dep$. By definition, $\coll{C}_\dep$ is finitely axiomatizable by $(1,0,2)$-dexrs. Therefore, $\coll{C}_\dep$ is diagrammatically $(1,0,2)$-compatible (by Theorem~\ref{the:dexr-refined-characterization}) and $(1,0)$-local (by Proposition~\ref{pro:compatibility-implies-locality}). Moreover, by Proposition~\ref{pro:compatibility-implies-locality}, we get that every collection of structures that is finitely axiomatizable by $(1,0,1)$-dexrs is also $(1,0)$-local. However, by exploiting our characterization, we can show that $\coll{C}_\dep$ is not finitely axiomatizable by $(1,0,1)$-dexrs.

%% file: linear-ders.tex
\section{Finite Axiomatizability by Guarded-based Disjunctive Existential Rules }\label{sec:guarded-based-dexrs}

The goal here is to establish a result analogous to Theorem~\ref{the:dexr-refined-characterization} for the two main members of the guarded family of dexrs:
\begin{itemize}
	\item A dexr is {\em linear} if it has at most one atom in its body.
	\item A dexr is {\em guarded} if its body is empty or has an atom that mentions all the universally quantified variables.
\end{itemize}
Interestingly, to achieve the desired characterizations for the above classes of dexrs, we simply need to replace the diagrammatic compatibility property in Theorem~\ref{the:dexr-refined-characterization} with a refined version of it that takes into account the syntactic property underlying linear and guarded dexrs. We start our analysis with linear dexrs, and then proceed with guarded dexrs.

\subsection{Linear Disjunctive Existential Rules}

A structure $J$ is called {\em linear} if $|\facts{J}| \leq 1$. Consider a collection $\coll{C}$ of $\ins{S}$-structures and an $\ins{S}$-structure $I$. For integers $n,m,\ell\geq 0$, we say that $\coll{C}$ is {\em linear-diagrammatically $(n,m,\ell)$-compatible} with $I$ if, for every linear structure $K \subseteq I$ with $\adom{K} = \aadom{K}$ and $|\adom{K}| \leq n$, and every $(m,\ell)$-diagram $\Delta_{K,G}^{I}$ of $K$ relative to $I$, where $G \subseteq N_{K,m}^{I}$, there exists $J \in \coll{C}$ such that $J \models \Delta_{K,G}^{I}$.

\begin{definition}\label{def:linear-diagrammatic-compatibility}
	A collection $\coll{C}$ of $\ins{S}$-structures is {\em linear-diagrammatically $(n,m,\ell)$-compatible}, for $n,m,\ell \geq 0$, if, for every $\ins{S}$-structure $I$, $\coll{C}$ is linear-diagrammatically $(n,m,\ell)$-compatible with $I$ implies $I \in \coll{C}$. \hfill\markfull
\end{definition}

It is important to observe that linear-diagrammatic compatibility implies diagrammatic compatibility as
this will be crucial for obtaining our main characterization.

\begin{lemma}\label{lem:linear-compatibility-to-compatibility}
	Consider a collection $\coll{C}$ of $\ins{S}$-structures that is linear-diagrammatically $(n,m,\ell)$-compatible, for integers $n,m \geq 0$, with $n+m>0$, and $\ell>0$. It holds that $\coll{C}$ is diagrammatically $(n,m,\ell)$-compatible.
\end{lemma}

The characterization of interest for linear dexrs follows:

\begin{theorem}\label{the:linear-dexr-refined-characterization}
	For a collection $\coll{C}$ of structures and $n,m \geq 0$, with $n+m > 0$, and $\ell > 0 $, the following are equivalent:
	\begin{enumerate}
		\item $\coll{C}$ is finitely axiomatizable by linear $(n,m,\ell)$-dexrs.
		\item $\coll{C}$ is critical, closed under repairable direct products, and linear-diagrammatically $(n,m,\ell)$-compatible.
	\end{enumerate}
\end{theorem}

To establish the above characterization, we first show a technical lemma, called {\em Linearization Lemma}, which is interesting in its own right as it characterizes when a collection of structures that is finitely axiomatizable by dexrs is finitely axiomatizable by linear dexrs. To this end, the property of linear-diagrammatic compatibility plays a central role.

\begin{lemma}\label{lem:linearization}
	Consider a collection $\coll{C}$ of structures that is finitely axiomatizable by $(n,m,\ell)$-dexrs, for integers $n,m \geq 0$, with $n+m >0$, and $\ell >0$. For every integer $\ell' > 0$, the following are equivalent: 
	\begin{enumerate}
		\item $\coll{C}$ is finitely axiomatizable by linear $(n,m,\ell')$-dexrs.
		\item $\coll{C}$ is linear-diagrammatically $(n,m,\ell')$-compatible.
	\end{enumerate}
\end{lemma}

The proof of the direction $(1) \Rightarrow (2)$ of the Linearization Lemma is analogous to the proof of Lemma~\ref{lem:dexr-locality}, whereas the direction $(2) \Rightarrow (1)$ is shown by following the same strategy as in the proof of the direction $(2) \Rightarrow (1)$ of Theorem~\ref{the:dexr-refined-characterization}.  
Having the Linearization Lemma in place, it is now not difficult to prove the characterization provided by Theorem~\ref{the:linear-dexr-refined-characterization}:

\begin{proof}[Proof of Theorem~\ref{the:linear-dexr-refined-characterization}]
The direction $(1) \Rightarrow (2)$ follows from Lemmas~\ref{lem:dexrs-critical}, \ref{lem:dexrs-closed-under-rpds}, and \ref{lem:linearization} (direction $(1) \Rightarrow (2)$).
For $(2) \Rightarrow (1)$, since $\coll{C}$ is linear-diagrammatically $(n,m,\ell)$-compatible, by Lemma~\ref{lem:linear-compatibility-to-compatibility} we get that $\coll{C}$ is also diagrammatically $(n,m,\ell)$-compatible. Since $\coll{C}$ is critical and closed under repairable direct products, we get from Theorem~\ref{the:dexr-refined-characterization} that $\coll{C}$ is finitely axiomatizable by $(n,m,\ell)$-dexrs. This allows us to apply Lemma~\ref{lem:linearization} (direction $(2) \Rightarrow (1)$), and get that $\coll{C}$ is finitely axiomatizable by linear $(n,m,\ell)$-dexrs, as needed.
\end{proof}

\subsection{Guarded Disjunctive Existential Rules}

Let us now proceed with guarded dexrs and perform a similar analysis as for linear dexrs. As one might suspect, the refined notion of diagrammatic compatibility with a structure $I$ is defined as diagrammatic compatibility with $I$ with the key difference that only guarded substructures $K$ of $I$ are considered. 
Formally, a structure $J$ is {\em guarded} if either $\facts{J} = \emptyset$, or there is $R(c_1,\ldots,c_r) \in \facts{J}$ such that $\aadom{J} = \{c_1,\ldots,c_r\}$.
Consider now a collection $\coll{C}$ of structures over a schema $\ins{S}$ and an $\ins{S}$-structure $I$. For integers $n,m,\ell\geq 0$, we say that $\coll{C}$ is {\em guarded-diagrammatically $(n,m,\ell)$-compatible} with $I$ if, for every guarded substructure $K$ of $I$ with $\adom{K} = \aadom{K}$ and $|\adom{K}| \leq n$, and every $(m,\ell)$-diagram $\Delta_{K,G}^{I}$ of $K$ relative to $I$, where $G \subseteq N_{K,m}^{I}$, there exists $J \in \coll{C}$ such that $J \models \Delta_{K,G}^{I}$.
The refined model-theoretic property follows.

\begin{definition}\label{def:guarded-diagrammatic-compatibility}
	A collection $\coll{C}$ of $\ins{S}$-structures is {\em guarded-diagrammatically $(n,m,\ell)$-compatible}, for $n,m,\ell \geq 0$, if, for every $\ins{S}$-structure $I$, $\coll{C}$ is guarded-diagrammatically $(n,m,\ell)$-compatible with $I$ implies $I \in \coll{C}$. \hfill\markfull
\end{definition}

As for linear-diagrammatic compatibility, it is straightforward to show the following, which will be used later:

\begin{lemma}\label{lem:guarded-compatibility-to-compatibility}
	Consider a collection $\coll{C}$ of $\ins{S}$-structures that is guarded-diagrammatically $(n,m,\ell)$-compatible, for integers $n,m \geq 0$, with $n+m>0$, and $\ell>0$. It holds that $\coll{C}$ is diagrammatically $(n,m,\ell)$-compatible.
\end{lemma}

The characterization of interest for guarded dexrs follows:

\begin{theorem}\label{the:guarded-dexr-refined-characterization}
	For a collection $\coll{C}$ of structures and $n,m \geq 0$, with $n+m > 0$, and $\ell > 0 $, the following are equivalent:
	\begin{enumerate}
		\item $\coll{C}$ is finitely axiomatizable by guarded $(n,m,\ell)$-dexrs.
		\item $\coll{C}$ is critical, closed under repairable direct products, and guarded-diagrammatically $(n,m,\ell)$-compatible.
	\end{enumerate}
\end{theorem}

To establish the above characterization, we first show a technical lemma in the spirit of the Linearization Lemma, called {\em Guardedization Lemma}, which characterizes when a collection of structures that is finitely axiomatizable by dexrs is finitely axiomatizable by guarded dexrs. To this end, guarded-diagrammatic compatibility plays a crucial role.

\begin{lemma}\label{lem:guardedization}
	Consider a collection $\coll{C}$ of structures that is finitely axiomatizable by $(n,m,\ell)$-dexrs, for integers $n,m \geq 0$, with $n+m >0$, and $\ell >0$. For every integer $\ell' > 0$, the following are equivalent:
	\begin{enumerate}
		\item $\coll{C}$ is finitely axiomatizable by guarded $(n,m,\ell')$-dexrs.
		\item $C$ is guarded-diagrammatically $(n,m,\ell')$-compatible.
	\end{enumerate}
\end{lemma}

The proof of the direction $(1) \Rightarrow (2)$ of the Guardedization Lemma is analogous to the proof of Lemma~\ref{lem:dexr-locality}, whereas the direction $(2) \Rightarrow (1)$ is shown by following the same strategy as in the proof of the direction $(2) \Rightarrow (1)$ of Theorem~\ref{the:dexr-refined-characterization}.  
The Guardedization Lemma allows us to prove the characterization provided by Theorem~\ref{the:guarded-dexr-refined-characterization}:

\begin{proof}[Proof of Theorem~\ref{the:guarded-dexr-refined-characterization}]
	The direction $(1) \Rightarrow (2)$ follows from Lemmas~\ref{lem:dexrs-critical}, \ref{lem:dexrs-closed-under-rpds}, and \ref{lem:guardedization} (direction $(1) \Rightarrow (2)$).
	For $(2) \Rightarrow (1)$, since $\coll{C}$ is guarded-diagrammatically $(n,m,\ell)$-compatible, by Lemma~\ref{lem:guarded-compatibility-to-compatibility} we get that $\coll{C}$ is also diagrammatically $(n,m,\ell)$-compatible. Since $\coll{C}$ is critical and closed under repairable direct products, we get from Theorem~\ref{the:dexr-refined-characterization} that $\coll{C}$ is finitely axiomatizable by $(n,m,\ell)$-dexrs. Thus, we can apply Lemma~\ref{lem:guardedization} (direction $(2) \Rightarrow (1)$), and get that $\coll{C}$ is finitely axiomatizable by guarded $(n,m,\ell)$-dexrs.
\end{proof}

%% file: rewritability.tex
\section{From Guarded to Linear Disjunctive Existential Rules}\label{sec:rewritability}

The goal of this last section is to understand whether our new diagrammatic compatibility property can be used to solve the non-trivial problem of rewriting a set of guarded dexrs into an equivalent set of dexrs that falls in the weaker class of linear dexrs rules. Formally, we are interested in the following algorithmic problem:

\medskip

\begin{center}
	\fbox{\begin{tabular}{ll}
			{\small PROBLEM} : & $\mathsf{G\text{-}to\text{-}L}$
			\\
			{\small INPUT} : & A finite set $\dep$ of guarded dexrs.
			\\
			{\small OUTPUT} : & A finite set $\dep'$ of linear dexrs such that\\
			& $\dep \equiv \dep'$, if one exists; otherwise, $\mathsf{fail}$.
	\end{tabular}}
\end{center}

\medskip

\noindent Our goal is to show that:

\begin{theorem}\label{the:rewritability}
	$\mathsf{G\text{-}to\text{-}L}$ is computable in elementary time.
\end{theorem}

\subsection{Bounded Linearization Lemma} 

To obtain an algorithm that solves our rewritability problem, we need a stronger version of the Linearization Lemma, established in the previous section (see Lemma~\ref{lem:linearization}), that allows us to bound the number of variables and the number of head-disjuncts occurring in the linear dexrs of the equivalent set, which in turn allows us to focus on finitely many linear dexrs. More precisely, we need a result that allows us to conclude the following: given a set $\dep$ of dexrs, there is a set $\dep'$ of linear dexrs that is equivalent to $\dep$ iff there is one consisting of linear dexrs with a bounded number of variables and a bounded number of head-disjuncts. This is achieved by the following result, dubbed {\em Bounded Linearization Lemma}. For a schema $\ins{S}$, we write $\ar{\ins{S}}$ for the maximum arity over all predicates of $\ins{S}$, that is, the integer $\max_{R \in \ins{S}} \ar{R}$.

\begin{lemma}\label{lem:boudned-linearization}
	Consider a collection $\coll{C}$ of structures over a schema $\ins{S}$ that is finitely axiomatizable by $(n,m,\ell)$-dexrs, for integers $n,m \geq 0$, with $n+m >0$, and $\ell >0$. Let $\ell' = \ell \cdot |\ins{S}| \cdot (n+m+1)^{m \cdot \ar{\ins{S}}}$. The following are equivalent: 
	\begin{enumerate}
		\item $\coll{C}$ is finitely axiomatizable by linear dexrs.
		\item $\coll{C}$ is linear-diagrammatically $(n,m,\ell')$-compatible.
	\end{enumerate}
\end{lemma}

\medskip
\noindent  \paragraph{Remark.} To establish a Bounded Linearization Lemma, one could simply let $\ell' = 2^{|\ins{S}| \cdot (n+m)^{\ar{\ins{S}}}}$. Indeed, with $n+m$ variables, we can construct at most $|\ins{S}| \cdot (n+m)^{\ar{\ins{S}}}$ atoms that mention a predicate of $\ins{S}$ and any subset of those atoms may give rise to a conjunction of atoms that can appear in a head-disjunct. 
However, we would like to obtain an optimal bound on the number of head-disjuncts, whereas the naive double-exponential bound discussed above is clearly suboptimal. The Bounded Linearization Lemma established above provides an exponential bound, which significantly improves the naive bound, by exploiting the fact that the collection $\coll{C}$ of $\ins{S}$-structures is finitely axiomatizable by dexrs with at most $\ell$ head-disjuncts. The question whether we can establish a stronger version of the Bounded Linearization Lemma, which provides a polynomial bound on the number of disjuncts, remains an interesting open problem.

\begin{algorithm}[t]
	\KwIn{A set $\dep$ of guarded $(n,m,\ell)$-dexrs over $\ins{S}$, for $n,m \geq 0$, with $n+m > 0$, and $\ell > 0$.}
	\KwOut{A set $\dep'$ of linear $(n,m,\ell')$-dexrs, where $\ell' = \ell \cdot |\ins{S}| \cdot (n+m)^{\ar{\ins{S}} \cdot (n+1)}$, such that $\dep \equiv \dep'$, if one exists; otherwise, $\mathsf{fail}$.}
	\vspace{2mm}
	
	$\ell' := \ell \cdot |\ins{S}| \cdot (n+m)^{\ar{\ins{S}} \cdot (n+1)}$\\
	$\dep' := \{\sigma \mid \sigma \text{ is a linear } (n,m,\ell')\text{-dexr over } \ins{S} \text{ and } \dep \models \sigma\}$\\
	
	\uIf{$\dep' \neq \emptyset$ \text{\rm and} $\dep' \models \dep$}
	{
		\textbf{return} $\dep'$
	}\Else{
		\textbf{return} $\mathsf{fail}$
	}
	\caption{\textsf{Rewrite}}
	\label{alg:guarded-to-linear}
\end{algorithm}

\subsection{The Rewriting Algorithm} 

Let $\dep$ be a finite set of guarded dexrs over $\ins{S}$. Clearly, there are integers $n,m \geq 0$, with $n+m > 0$, and $\ell > 0$, such that $\dep$ consists only of $(n,m,\ell)$-dexrs. By the Bounded Linearization Lemma, we get that the following are equivalent:
\begin{itemize}
	\item There is a finite set $\dep'$ of linear dexrs over $\ins{S}$ with $\dep \equiv \dep'$.
	
	\item There is a finite set $\dep'$ of linear $(n,m,\ell')$-dexrs, where $\ell' = \ell \cdot |\ins{S}| \cdot (n+m+1)^{m \cdot \ar{\ins{S}}}$, over $\ins{S}$ with $\dep \equiv \dep'$.
\end{itemize}
This means that, even though there are infinitely many finite sets of linear dexrs over $\ins{S}$, it suffices to search only for linear dexrs over $\ins{S}$ that mention at most $n$ universally quantified variables, at most $m$ existentially quantified variables, and at most $\ell'$ head-disjuncts, which are finitely many, to find a set $\dep'$ that is equivalent to $\dep$.
This leads to the simple algorithm depicted in Algorithm~\ref{alg:guarded-to-linear}. It first collects in $\dep'$ all the linear $(n,m,\ell')$-dexrs over $\ins{S}$ that are entailed by the input set $\dep$ of guadred dexrs, and then checks whether $\dep'$ is non-empty and entails $\dep$; the latter is actually done by checking whether $\dep' \models \sigma$, for each $\sigma \in \dep$.
We proceed to show that $\mathsf{Rewrite}$ runs in elementary time, which will imply Theorem~\ref{the:rewritability}.

We first observe that the number of linear $(n,m,\ell')$-dexrs over the schema $\ins{S}$ is bounded by the integer
\[
\underbrace{|\ins{S}| \cdot n^{\ar{\ins{S}}}}_{\geq \text{ \#~of linear bodies }} \cdot \underbrace{\sum_{i=1}^{\ell'} \binom{H}{i}}_{\geq \text{ \#~of heads}},
\]
where $H = 2^{|\ins{S}| \cdot (n+m)^{\ar{\ins{S}}}}$. Therefore, one can enumerate all the linear $(n,m,\ell')$-dexrs over $\ins{S}$ in elementary time (in fact, in triple-exponential time).
It remains to analyze the complexity of deciding whether a set of guarded dexrs entails a linear dexr (needed in the construction of $\dep'$), and the complexity of deciding whether a set of linear dexrs entails a guarded dexr (needed for checking whether $\dep' \models \dep$).
Given a set $\dep$ of dexrs and a single dexr $\sigma$ of the usual form $\phi(\bar x,\bar y) \ra  \bigvee_{i=1} ^{k} \exists \bar z_i \, \psi_i(\bar x_i,\bar z_i)$, it is not difficult to show that the following statements are equivalent:
\begin{enumerate}
	\item $\dep \models \sigma$.
	\item $\phi(\rho(\bar x),\rho(\bar y)) \wedge \dep \models \exists \bar z_i \, \psi_i(\rho(\bar x_i),\bar z_i)$, for some $i \in [k]$, where $\rho$ is a renaming function that replaces each variable $u$ in $\phi(\bar x,\bar y)$ with a new constant $\rho(u)$.
\end{enumerate}
The problem of deciding whether $\phi(\rho(\bar x),\rho(\bar y)) \wedge \dep \models \exists \bar z_i \, \psi_i(\rho(\bar x_i),\bar z_i)$, for some $i \in [k]$, if $\dep$ is a set of guarded (and thus, linear) dexrs is decidable in elementary time (in particular, in double-exponential time); the latter is immediately inherited from~\cite{BoMMP16} that analyzes the problem of conjunctive query answering under guarded dexrs.
Putting everything together, we get that the algorithm $\mathsf{Rewrite}$ runs in elementary time and Theorem~\ref{the:rewritability} follows.

At this point, let us stress that a more refined complexity analysis allows us to conclude that $\mathsf{G\text{-}to\text{-}L}$ is computable in triple-exponential time, whereas the best that we can hope for is double-exponential time since already the problem of deciding whether a set of guarded existential rules (i.e., dexrs with only one head-disjunct) can be rewritten as an equivalent set of linear existential rules is 2EXPTIME-hard~\cite{CoKP21}. More precisely, from the proof of the Bounded Linearization Lemma, we can conclude that whenever the collection $\coll{C}$ of $\ins{S}$-structures is finitely axiomatizable by linear dexrs where each head-disjunct mentions at most $p>0$ atoms, then it is finitely axiomatizable by linear $(n,m,\ell')$-dexrs with the same bound on the number of atoms in each head-disjunct. Therefore, the integer $H$ in the above analysis can be set to
\[
\sum_{i = 1}^{p} \binom{|\ins{S}| \cdot (n+m)^{\ar{\ins{S}}}}{i}\ \leq\ \sum_{i = 1}^{p} \left(|\ins{S}| \cdot (n+m)^{\ar{\ins{S}}}\right)^i,
\]
which implies that we need to consider double-exponentially many linear dexrs.
This in turn allows us to argue, by using results on the complexity of the satisfiability problem for the guarded fragment of first-order logic~\cite{Grad99}, that our rewriting algorithm runs in triple-exponential time.

\medskip
\noindent  \paragraph{Remark.} As said above, it remains open whether a stronger version of the Bounded Linearization Lemma that provides a polynomial bound on the number of head-disjuncts can be established. Let us remark that having such a stronger version of the Bounded Linerization Lemma in place, we can show that our problem $\mathsf{G\text{-}to\text{-}L}$ is computable in double-exponential time, which is the best that we can hope for.

%% file: conclusions.tex
\section{Future Work}\label{sec:conclusions}

We would like to perform a similar analysis that goes beyond finite axiomatizability by dexrs. In particular, we are planning to consider finite axiomatizability by dexrs, disjunctive equality rules, and denial constraints.
Moreover, it is interesting to prove a stronger version of the Bounded Linearization Lemma (see the last remark of Section~\ref{sec:rewritability}), which will allow us to establish that $\mathsf{G\text{-}to\text{-}L}$ is computable in double-exponential time, which is worst-case optimal.

\section*{Acknowledgements}

This work was funded by the European Union - Next Generation EU under the MUR PRIN-PNRR grant P2022KHTX7 ``DISTORT''.

%% file: appendix.tex
\section{Proofs from Section~\ref{sec:properties}}

\subsection*{Proof of Lemma~\ref{lem:dexrs-critical}}

We need to show that, for every $\kappa>0$, $\coll{C}$ contains a $\kappa$-critical structure. 
Since $\coll{C}$ is finitely axiomatizable by dexrs, there is a finite set $\dep$ of dexrs such that $I \in \coll{C}$ iff $I \models \dep$. 
Fix an arbitrary integer $\kappa>0$ and a $\kappa$-critical $\ins{S}$-structure $I$. It suffices to show that $I \models \dep$, which in turn implies that $I \in \coll{C}$, and thus, $\coll{C}$ is critical.
Consider a dexr $\sigma \in \dep$ of the form $\phi(\bar x,\bar y) \ra \bigvee_{i=1}^{k} \exists \bar z_i\, \psi_i(\bar x_i,\bar z_i)$. 
Assume first that $\phi(\bar x,\bar y)$ is empty. It is easy to see that the function $h : \bar z_i \ra \adom{I}$ such that $h(z) = c$ for each $z \in \bar z_i$, for some arbitrary $i \in [k]$ and $c  \in \adom{I}$, is such that $h(\psi_i(\bar z_i)) \subseteq \facts{I}$, and thus, $I \models \sigma$.
Assume now that $\phi(\bar x,\bar y)$ is non-empty. It is clear that there exists a function $h : \bar x \cup \bar y \ra \adom{I}$ such that $h(\phi(\bar x,\bar y)) \subseteq \facts{I}$. Let $c \in \adom{I}$ and consider the extension $h'$ of $h$ such that $h(z) = c$ for each $z \in \bar z_i$ for some arbitrary $i \in [k]$. Since $I$ is $\kappa$-critical, it is straightforward to verify that $h'(\psi_i(\bar x_i,\bar z_i))\subseteq \facts{I}$, which implies that $I \models \sigma$, and the claim follows.

\subsection*{Proof of Lemma~\ref{lem:dexrs-closed-under-rpds}}

We first need to introduce the chase procedure, a very useful algorithmic tool when reasoning with rule-like sentences such a dexrs, and recall its main properties. Once the chase procedure for dexrs is in place, we will then proceed with the proof of Lemma~\ref{lem:dexrs-closed-under-rpds}.

\medskip

\noindent
\paragraph{The Chase Procedure.} 
We first define the notion of chase application, which is the building block of the chase procedure, and then introduce the procedure itself.

A trigger for a set $\dep$ of dexrs over a schema $\ins{S}$ on an $\ins{S}$-structure $I$ is a pair $(\sigma,h)$, where $\sigma$ is a dexr from $\dep$ of the form $\phi(\bar x,\bar y) \ra \bigvee_{i=1}^{k} \exists \bar z_i\, \psi_i(\bar x_i,\bar z_i)$, and $h = \epsilon$ if $\phi(\bar x,\bar y)$ is empty; otherwise, if $\phi(\bar x,\bar y)$ is non-empty, then $h$ is a function from $\bar x \cup \bar y$ to $\adom{I}$ such that $h(\phi(\bar x,\bar y)) \subseteq \facts{I}$. Such a trigger $(\sigma,h)$ is {\em active} if the following holds: if $h = \epsilon$, then there is no integer $i \in [n]$ and a function $h' : \bar z_i \ra \adom{I}$ such that $h(\psi_i(\bar z_i)) \subseteq \facts{I}$; otherwise, there is no $i \in [n]$ and an extension $h'$ of $h$ such that $h'(\psi_i(\bar x_i,\bar z_i)) \subseteq \facts{I}$. In simple words, $(\sigma,h)$ is active if it witnesses a violation of $\sigma$.
An {\em application} of a trigger $(\sigma,\epsilon)$ to $I$ returns the set of $\ins{S}$-structures $\{I_1,\ldots,I_k\}$ such that, for each $i \in [k]$, $\adom{I_i} = \adom{I} \cup \{h(z) \mid z \in \bar z_i\}$, and $\facts{I_i} = \facts{I} \cup h(\psi_i(\bar z_i))$, where $h$ is a function from $\bar z_i \ra \ins{C}$ such that (i) for each existentially quantified variable $z$ in $\bar z_i$, $h(\bar z_i)$ is a new constant of $\ins{C}$ not occurring in $I$, and (ii) for distinct variables $z$ and $w$ in $\bar z_i$, $h(z) \neq h(w)$. 
Analogously, an {\em application} of a trigger $(\sigma,h)$ to $I$, where $h \neq \epsilon$, returns the set of $\ins{S}$-structures $\{I_1,\ldots,I_k\}$ such that, for each $i \in [k]$, $\adom{I_i} = \adom{I} \cup \{h'(z) \mid z \in \bar z_i\}$, and $\facts{I_i} = \facts{I} \cup h'(\psi_i(\bar x_i,\bar z_i))$, where $h'$ extends $h$ in such a way that (i) for each existentially quantified variable $z$ in $\bar z_i$, $h'(\bar z_i)$ is a new constant of $\ins{C}$ not occurring in $I$, and (ii) for distinct variables $z$ and $w$ in $\bar z_i$, $h'(z) \neq h_i(w)$.
The application of a trigger $(\sigma,h)$ is denoted $I \langle\sigma,h\rangle \{I_1,\ldots,I_k\}$.

The main idea of the chase procedure is, starting from $I$, to exhaustively apply active triggers for $\dep$  until a fixpoint is reached, which will result to a (possibly infinite) set of models of $\dep$ that include $I$ with a key property known as universality. This is formalized via disjunctive chase trees.
A {\em disjunctive chase tree} of $I$ and $\dep$ is a (possibly infinite) rooted node-labeled tree $T = (V,E,\lambda)$, where $V$ is the set of nodes, $E$ is the set of edges, and $\lambda$ is the labeling function that assigns to each node of $V$ an $\ins{S}$-structure, such that:

\begin{enumerate}
	\item if $v \in V$ is the root, then $\lambda(v) = I$, and
	\item If $v \in V$ is a non-leaf node with $u_1,\ldots,u_k$, for $k>0$, being its children, then there exists an active trigger $(\sigma,h)$ for $\dep$ on $\lambda(v)$ such that $\lambda(v) \langle \sigma,h \rangle \{\lambda(u_1),\ldots,\lambda(u_k)\}$.
\end{enumerate}

\noindent A {\em maximal path} of $T$ is a path of $T$ that starts with the root of $T$, and either ends in a leaf of $T$ or is infinite. We write $\mathsf{paths}(T)$ for the set of all maximal paths of $T$. Furthermore, for a maximal path $\pi$ of $T$, we define its result, denoted $\mathsf{result}(\pi)$, as follows: if $\pi = v_1,\ldots,v_n$ is finite, then $\mathsf{result}(\pi)$ is the structure $\lambda(v_n)$; otherwise, if $\pi = v_1,v_2,\ldots$ is infinite, then $\mathsf{result}(\pi)$ is defined as the structure $\bigcup_{i>0} \lambda(v_i)$;\footnote{The union of structures is defined in the obvious way.} the latter is well-defined since, for each $i,j>0$ with $i < j$, $\lambda(v_i) \subseteq \lambda(v_j)$.
A maximal path $\pi = v_1,v_2,\ldots$ of $T$ that is infinite is called {\em fair} if, for each $i>0$ and active trigger $(\sigma,h)$ for $\dep$ on $\lambda(v_i)$, there is $j>i$ such that $(\sigma,h)$ is not active on $\lambda(v_j)$. The disjunctive chase tree $T$ is {\em fair} if every path of $\mathsf{paths}(T)$ is either finite, or infinite and fair.
The result of the disjunctive chase tree $T$ of $I$ and $\dep$, denoted $\chase_{T}(I,\dep)$, is the (possibly infinite) set of structures $\{\result{\pi} \mid \pi \in \paths{T}\}$.
The well-known key properties of the chase procedure are given below. 
%These are folklore results that, to the best of our knowledge, have not been explicitly shown in the literature. 
%For the sake of completeness, we provide the formal proofs in the appendix.

\begin{proposition}\label{pro:disjucntive-chase}
	Consider an $\ins{S}$-structure $I$ and a set $\dep$ of dexrs over $\ins{S}$. For every disjunctive chase tree $T$ of $I$ and $\dep$, the following hold:
	\begin{enumerate}
		\item For every $\ins{S}$-structure $J \in \chase_T(I,\dep)$, $J \models \dep$.
		\item For every $\ins{S}$-structure $J$ with $h: I \ra J$ and $J \models \dep$, there exist an $\ins{S}$-structure $K \in \chase_T(I,\dep)$ and an extension $h'$ of $h$ such that $h' : K \ra J$.
	\end{enumerate}
\end{proposition}

\medskip

\noindent
\paragraph{The Proof of Lemma~\ref{lem:dexrs-closed-under-rpds}.} Let $I,J \in \coll{C}$. We proceed to show that there is a repairable direct product of $I$ and $J$ in $\coll{C}$. Since $\coll{C}$ is finitely axiomatizable by dexrs, there is a finite set $\dep$ of dexrs such that $K \in \coll{C}$ iff $K \models \dep$. 
Since the projective homomorphism for $I \otimes J$ is a homomorphism from $I \otimes J$ to $I$, by Proposition~\ref{pro:disjucntive-chase}, there is $L \in \chase_T(I \otimes J,\dep)$, where $T$ is a disjunctive chase tree of $I \otimes J$ and $\dep$, and an extension $h$ of $\pi_{I \otimes J}$ that is a homomorphism from $I \otimes J$ to $L$. By construction, $I \otimes J \subseteq L$. Thus, $L$ is a repairable direct product of $I$ and $J$. By Proposition~\ref{pro:disjucntive-chase}, $L \models \dep$, which implies that $L \in \coll{C}$, and the claim follows.

%%%%%%%%%%%

\section{Proofs from Section~\ref{sec:dexrs}}

\subsection*{Proof of Claim~\ref{cla:equivalent-dd}}

Recall that $\Phi_{K,G}^{I}(\bar x)$ is obtained from $\Delta_{K,G}^{I}$ by renaming each $c \in \adom{K}$ to a new variable $x_c$; let $\rho$ be the renaming function, that is, $\rho(c) = x_c$ for each $c \in \adom{K}$, and $\rho(y) = y$ for each variable $y$ occurring in $\Delta_{K,G}^{I}$. Thus, $\Phi_{K,G}^{I}(\bar x)$ is of the form $\Psi_1 \wedge \Psi_2$, where $\Psi_1$ is
\[
\bigwedge_{\alpha \in \facts{K}} \rho(\alpha)
\]
and $\Psi_2$ is the formula
\[
\bigwedge_{\substack{c,d \in \adom{K}, \\ c \neq d}} \neg (\rho(c)=\rho(d))\ \wedge \bigwedge_{\gamma(\bar y) \in G} \neg (\exists \bar y \, \rho(\gamma(\bar y))).
\]
Note that $\Psi_1$ and $\Psi_2$ might be empty (i.e., they have no conjuncts). In particular, $\Psi_1$ is empty if $K$ is empty, whereas $\Psi_2$ is empty if $I$ is $1$-critical, which means that $|\adom{K}|=1$ and $I \models \exists \bar y \, \gamma(\bar y)$ for each $\gamma(\bar y) \in C_{K,m}$, and thus, $G = \emptyset$.
We consider all the cases where $\Psi_1$ and $\Psi_2$ are empty or not. 
\begin{enumerate}
	\item We first assume that both $\Psi_1$ and $\Psi_2$ are empty. In this case, $\neg \exists \bar x \,\Phi_{K,G}^{I}(\bar x)$ is the truth constant $\mathsf{false}$, i.e., a contradiction, and it trivially occurs in $\class{DD}_{n,m,\ell}^{\ins{S}}$.
		
	\item We now assume that $\Psi_1$ is non-empty and $\Psi_2$ is empty. In this case, $\neg \exists \bar x \,\Phi_{K,G}^{I}(\bar x)$ is equivalent to the dd $\delta = \forall \bar x (\Psi_1 \ra \mathsf{false})$. Since $|\adom{K}| \leq n$, we conclude that $\bar x$ consists of at most $n$ variables, and thus, $\delta \in \class{DD}_{n,m,\ell}^{\ins{S}}$.
		
	\item Consider now the case where $\Psi_1$ is either empty or not and $\Psi_2$ is non-empty. Then, $\neg \exists \bar x \,\Phi_{K,G}^{I}(\bar x)$ is equivalent to the dd $\delta = \forall \bar x (\Psi_1 \ra \Xi)$, where 
	\[
	\Xi\ =\ \bigvee_{\substack{c,d \in \adom{K}, \\ c \neq d}} \rho(c) = \rho(d) \vee \bigvee_{\gamma(\bar y) \in G} \exists \bar y \, \rho(\gamma(\bar y)).
	\]
	It remains to show that $\delta \in \mathsf{DD}_{n,m,\ell}^{\ins{S}}$. To this end, we need to show the following three statements: (i) each variable in $\Xi$ is either existentially quantified or appears in $\Psi_1$,  (ii) $\delta$ mentions at most $n$ universally quantified variables and at most $m$ existentially quantified variables, and (iii) at most $\ell$ disjuncts of $\Xi$ are a non-empty conjunction of atoms (i.e., not an equality).	
	Statement (i) holds since, by hypothesis, $\adom{K} = \aadom{K}$. Concerning statement (ii), $\delta$ has at most $n$ universally quantified variables since, by hypothesis, $|\aadom{K}| \leq n$ and $\delta$ has $m$ existentially quantified variables because so does $\neg \forall \bar x \Phi_{K,G}^I(\bar x)$ by the construction of $\Phi_{K,G}^I(\bar x)$. Finally, statement (iii) holds since $G$ consists of at most $\ell$ formulas.
\end{enumerate}	
This completes the proof of the claim.

\subsection*{Diagramatic Compatibility vs Locality}

We first recall the property of locality introduced in~\cite{CoKP21} and then proceed to establish (i) Proposition~\ref{pro:compatibility-implies-locality}, and (ii) the fact that $\coll{C}_\dep$ is not finitely axiomatizable by $(1,0,1)$-dexrs, where $\coll{C}_{\dep}$ is the set of models of $\dep$ consisting of the single dexr $R(x) \ra S(x) \vee T(x)$.

\medskip
\noindent \paragraph{Locality.} The property of locality relies on the notion of local embedding of a collection of structures in a structure.
Roughly, a collection $\coll{C}$ of structures over a schema $\ins{S}$ is locally embeddable in an $\ins{S}$-structure $I$ if, for every substructure $K$ of $I$ with a bounded number of active domain elements (i.e., domain elements that occur in $\facts{K}$), we can find a structure $J_K \in C$ such that every local neighbour of $K$ in $J_K$ (i.e., substructures of $J_K$ that contain $K$ and have a bounded number of additional active domain elements not in $\facts{K}$), can be embedded in $I$ while preserving $K$.
Now, a collection of structures $\coll{C}$ is local if, for every $\ins{S}$-structure $I$, $\coll{C}$ is locally embeddable in $I$ implies that $I$ belongs to $\coll{C}$.
We proceed to formalize the above description.

Consider an $\ins{S}$-structure $J$ and a finite set of constants $F \subseteq \aadom{J}$. For an integer $m \geq 0$, the {\em $m$-neighbourhood of $F$ in $J$} is the set of $\ins{S}$-structures
\[
\{J' \mid F \subseteq \aadom{J'},\ J' \preceq J~~\text{ and }~~|\aadom{J'}| \leq |F| + m\},
\]
i.e., all the substructures of $J$ such that their facts contain constants from $F$ and at most $m$ additional elements not occurring in $F$.
For an $\ins{S}$-structure $K \subseteq J$, the {\em $m$-neighbourhood of $K$ in $J$} is defined as the $m$-neighbourhood of $\aadom{K}$ in $J$, that is, all the substructures of $J$ that contain $K$ and their facts mention at most $m$ additional elements not occurring in the facts of $K$.

Let $\coll{C}$ be a collection of $\ins{S}$-structures and $I$ an $\ins{S}$-structure. For $n,m \geq 0$, we say that {\em $\coll{C}$ is $(n,m)$-locally embeddable in $I$} if, for every $K \preceq I$ with $|\aadom{K}| \leq n$, there is $J_K \in \coll{C}$ such that $K \subseteq J_K$, and for every $J'$ in the $m$-neighbourhood of $K$ in $J_K$, there is a function $h : \aadom{J'} \ra \aadom{I}$, which is the identity on $\aadom{K}$, such that $h(\facts{J'}) \subseteq \facts{I}$. The property of locality follows.

\begin{definition}\label{def:locality}
	A collection $\coll{C}$ of $\ins{S}$-structures is {\em $(n,m)$-local}, for $n,m \geq 0$, if, for every $\ins{S}$-structure $I$, the following holds: $\coll{C}$ is $(n,m)$-locally embeddable in $I$ implies $I \in \coll{C}$. \hfill\markfull
\end{definition}

\medskip
\noindent \paragraph{Proof of Proposition~\ref{pro:compatibility-implies-locality}.} Fix arbitrary integers $n,m \geq 0$, with $n+m > 0$, and $\ell > 0$. Assume that $\coll{C}$ is $(n,m)$-locally embeddable in a structure $I$. We need to show that $I \in \coll{C}$. To this end, it suffices to show that $\coll{C}$ is diagrammatically $(n,m,\ell)$-compatible with $I$, which in turn implies that $I \in \coll{C}$ since $\coll{C}$ is diagrammatically $(n,m,\ell)$-compatible.

Consider a structure $K \preceq I$ with $\adom{K} = \aadom{K}$ and $|\adom{K}| \leq n$, and an $(m,\ell)$-diagram $\Delta_{K,G}^{I}$ of $K$ relative to $I$, where $G \subseteq N_{K,m}^{I}$, of the usual form 
\[
\bigwedge_{\alpha \in \facts{K}} \alpha\ \wedge\ \bigwedge_{\substack{c,d \in \adom{K}, \\ c \neq d}} \neg (c=d)\ \wedge\ \\ \bigwedge_{\gamma(\bar y) \in G} \neg \left(\exists \bar y \, \gamma(\bar y)\right).
\]
We need to show that there is a structure $J \in \coll{C}$ such that $J \models \Delta_{K,G}^{I}$. 
Since $\coll{C}$ is $(n,m)$-locally embeddable in $I$, there is $J_K \in \coll{C}$ such that $K \subseteq J_K$, and for every $J'$ in the $m$-neighbourhood of $K$ in $J_K$, there is a function $h_{J'} : \aadom{J'} \ra \aadom{I}$, which is the identity on $\aadom{K}$, such that $h_{J'}(\facts{J'}) \subseteq \facts{I}$. We proceed to show that $J_K \models \Delta_{K,G}^{I}$, which will compete the proof.

By contradiction, assume that $J_K \not\models \Delta_{K,G}^{I}$. It is clear that $J_K \models \bigwedge_{\alpha \in \facts{K}} \alpha \wedge \bigwedge_{c,d \in \adom{K}, c \neq d} \neg (c=d)$. Therefore, there exists $\gamma(\bar y) \in G$ such that $J_K \models \exists \bar y \, \gamma(\bar y)$. This in turn implies that there exists a function $h : T_\gamma \ra \aadom{J_k}$, where $T_{\gamma(\bar y)}$ is the set of variables and constants occurring in $\gamma(\bar y)$, that is the identity on $\aadom{K}$ such that $h(\gamma(\bar y)) \subseteq \facts{J_K}$. Since $\gamma(\bar y)$ mentions at most $m$ variables, there exists $K'$ in the $m$-neighbourhood of $K$ in $J_K$ such that $h(\gamma(\bar y)) \subseteq \facts{K'}$. Let $\lambda = h_{K'} \circ h$. It is clear that $\lambda$ is a function from $T_{\gamma(\bar y)}$ to $\aadom{I}$, which is the identity on $\aadom{K}$, such that $\lambda(\gamma(\bar y)) \subseteq \facts{I}$. Hence, $I \models \exists y \, \gamma(\bar y)$. Therefore, $I \not\models \Delta_{K,G}^{I}$, which is a contradiction since, by Lemma~\ref{lem:sat-diagram}, $I \models \Delta_{K,G}^{I}$.

\medskip
\noindent \paragraph{$\coll{C}_\dep$ is Not Finitely Axiomatizable by $(1,0,1)$-dexrs.} We now proceed to show that $\coll{C}_\dep$ is not finitely axiomatizable by $(1,0,1)$-dexrs, where $\coll{C}_{\dep}$ is the set of models of $\dep$ consisting of the single dexr $R(x) \ra S(x) \vee T(x)$.
By definition, $\coll{C}_\dep$ is finitely axiomatizable by dexrs, and thus, by Lemmas~\ref{lem:dexrs-critical} and~\ref{lem:dexrs-closed-under-rpds}, $\coll{C}_\dep$ is critical and closed under repairable direct products. Consequently, by Theorem~\ref{the:dexr-refined-characterization}, to probe our claim it suffices to show that $\coll{C}_\dep$ is not diagrammatically $(1,0,1)$-compatible. To this end, we need to show that there exists a structure $I \not\in \coll{C}_\dep$ over the schema $\ins{S} = \{R,S,T\}$ such that $\coll{C}_\dep$ is diagrammatically $(1,0,1)$-compatible with $I$. We claim that such an $\ins{S}$-structure is the one with $\adom{I} = \{R(a)\}$ and $R^I = \{(a)\}$, $S^I = \emptyset$ and $T^I = \emptyset$, namely $\facts{I} = \{R(a)\}$. Clearly, $I \not\models \dep$, and thus, $I \not\in \coll{C}_\dep$. It remains to show that $\coll{C}_\dep$ is indeed diagrammatically $(1,0,1)$-compatible with $I$.
We observe that $I$ has only two substructures $K$ with $\adom{K} = \aadom{K}$ and $|\adom{K}| \leq 1$: (i) the empty substructure $K_\emptyset$, where $\adom{K_\emptyset} = \emptyset$ and $\facts{K_\emptyset} = \emptyset$, and (ii) $I$ itself.
We proceed to show that, for every $K \in \{K_\emptyset,I\}$, it holds that, for every $(0,1)$-diagram $\Delta_{K,G}^{I}$ of $K$ relative to $I$, where $G \subseteq N_{K,0}^{I}$, there is a structure $J \in \coll{C}_\dep$ such that $J \models \Delta_{K,G}^{I}$, which proves that $\coll{C}_\dep$ is diagrammatically $(1,0,1)$-compatible with $I$.

We first consider the simple case of the empty substructure $K_\emptyset$. It is clear that $N_{K_\emptyset,0}^{I} = \emptyset$, and thus, the only subset $G$ of $N_{K_\emptyset,0}^{I}$ with cardinality at most one is the empty set. Hence, the only $(0,1)$-diagram of $K$ relative to $I$ that we need to consider is $\Delta_{K_\emptyset,\emptyset}^{I}$ that is a tautology, and thus, it is trivially satisfied by every structure of $\coll{C}_\dep$. 

We now consider the case of $I$ itself. It is easy to see that 
\begin{multline*}
N_{I,0}^{I}\ =\ \{S(a),T(a),R(a) \wedge S(a), R(a) \wedge T(a),\\
S(a) \wedge T(a),R(a) \wedge S(a) \wedge T(a)\}.
\end{multline*}
Thus, we have the following seven subsets of $N_{I,0}^{I}$ of cardinality at most one:
\begin{eqnarray*}
	G_1 &=& \emptyset\\
	G_2 &=& \{S(a)\}\\
	G_3 &=& \{T(a)\}\\
	G_4 &=& \{R(a) \wedge S(a)\}\\
	G_5 &=& \{R(a) \wedge T(a)\}\\
	G_6 &=& \{S(a) \wedge T(a)\}\\
	G_7 &=& \{R(a) \wedge S(a) \wedge T(a)\}.
\end{eqnarray*}
It is also not difficult to verify that
\begin{eqnarray*}
	\Delta_{I,G_1}^{I} &=& R(a)\\
	\Delta_{I,G_2}^{I} &\equiv& \Delta_{I,G_4}^{I}\ \equiv\ R(a) \wedge \neg S(a)\\
	\Delta_{I,G_3}^{I} &\equiv& \Delta_{I,G_5}^{I}\ \equiv\ R(a) \wedge \neg T(a)\\
	\Delta_{I,G_6}^{I} &\equiv& \Delta_{I,G_7}^{I}\ \equiv\ (R(a) \wedge \neg S(a)) \vee (R(a) \wedge \neg T(a)).
\end{eqnarray*}
It is now straightforward to see that, for each $i \in [7]$, there is a structure $J_i \in \coll{C}_\dep$ such that $J_i \models \Delta_{I,G_i}^{I}$. In particular, for $i \in \{1,2,4,6,7\}$, $\Delta_{I,G_i}^{I}$ is satisfied by the structure $J' \in \coll{C}_\dep$ with $\adom{J'} = \{a\}$ and $\facts{J'} = \{R(a),T(a)\}$, and for $i \in [3,5]$, $\Delta_{I,G_i}^{I}$ is satisfied by the structure $J'' \in \coll{C}_\dep$ with $\adom{J'} = \{a\}$ and $\facts{J''} = \{R(a),S(a)\}$.

%%%%%%%%%%%%%%%%%

\section{Proofs from Section~\ref{sec:guarded-based-dexrs}}

\subsection*{Proof of Lemma~\ref{lem:linear-compatibility-to-compatibility}}

Consider an arbitrary $\ins{S}$-structure $I$ such that $\coll{C}$ is diagrammatically $(n,m,\ell)$-compatible with $I$. Thus, by definition, $\coll{C}$ is linear-diagrammatically $(n,m,\ell)$-compatible with $I$. Since $\coll{C}$ is linear-diagrammatically $(n,m,\ell)$-compatible, we get that $I \in \coll{C}$, which in turn implies that $\coll{C}$ is diagrammatically $(n,m,\ell)$-compatible, as needed.

\subsection*{Proof of Lemma~\ref{lem:linearization}}

\noindent
\paragraph{\underline{The Direction $(1) \Rightarrow (2)$}}

\smallskip

\noindent The proof is analogous to the proof of Lemma~\ref{lem:dexr-locality}, which we give here for completeness.
Let $\dep$ be a finite set of linear $(n,m,\ell')$-dexrs with $I \in \coll{C}$ iff $I \models \dep$. Consider a structure $I$ and assume that $\coll{C}$ is linear-diagrammatically $(n,m,\ell')$-compatible with $I$. We need to show that $I \in \coll{C}$, or, equivalently, $I \models \dep$.
Consider a linear dexr $\sigma \in \dep$ of the form $\phi(\bar x,\bar y) \ra \bigvee_{i=1}^{k} \exists \bar z_i\, \psi_i(\bar x_i,\bar z_i)$, where $\phi(\bar x,\bar y)$ is non-empty and consists of a single atom; the case where $\phi(\bar x,\bar y)$ is empty is treated similarly. 
Assume that there is a function $h : \bar x \cup \bar y \ra \adom{I}$ such that $h(\phi(\bar x, \bar y)) \subseteq \facts{I}$. We need to show that there exists $i \in [k]$ such that $\lambda(\psi_i(\bar x_i, \bar z_i)) \subseteq \facts{I}$ with $\lambda$ being an extension of $h$.
By contradiction, assume that such $i \in [k]$ does not exist. This means that, for every $i \in [k]$, $I \models \neg \exists \bar z_i \, \psi'(\bar z_i)$, where $\psi'(\bar z_i)$ is obtained from $\psi(\bar x_i,\bar z_i)$ by replacing each variable $x$ of $\bar x_i$ with $h(x)$.
Let $K$ be the structure with $\adom{K}$ being the set of constants occurring in $h(\phi(\bar x,\bar y))$ and $\facts{K} = \{h(\phi(\bar x,\bar y))\}$. Clearly, $K$ is a linear structure, since $\phi(\bar x,\bar y)$ consists of a single atom, where $\adom{K} = \aadom{K}$ and $|\aadom{K}| \leq n$; the latter holds since $\phi(\bar x,\bar y)$ mentions at most $n$ variables. 
Let $G = \{\psi'_i(\bar z_i)\}_{i \in [k]}$. Since, for each $i \in [k]$, $\psi_i(\bar x_i,\bar z_i)$ mentions at most $m$ existentially quantified variables and $k \leq \ell'$, it is easy to see that $G \subseteq N_{K,m}^{I}$ and $|G| \leq \ell'$. Therefore, $\Delta_{K,G}^{I}$, that is, the $G$-diagram of $K$ relative to $I$ is an $(m,\ell')$-diagram of $K$ relative to $I$. Since $\coll{C}$ is linear-diagrammatically $(n,m,\ell')$-compatible with $I$, we get that there exists $J \in \coll{C}$ that satisfies $\Delta_{K,G}^{I}$. The latter implies that $h(\phi(\bar x,\bar y)) \subseteq \facts{J}$, but there is no extension $\lambda$ of $h$ such that $\lambda(\psi(\bar x_i,\bar z_i)) \subseteq \facts{J}$ for some $i \in [k]$. Consequently, $J \not\models \sigma$, and thus, $J \not\models \dep$. But this contradicts the fact that $J \in \coll{C}$, which is equivalent to say that $J \models \dep$.

\medskip

\noindent
\paragraph{\underline{The Direction $(2) \Rightarrow (1)$}}

\smallskip

\noindent This direction is shown by following the same strategy as in the proof of the direction $(2) \Rightarrow (1)$ of Theorem~\ref{the:dexr-refined-characterization}.  
Consider a collection $\coll{C}$ of structures over a schema $\ins{S}$ that is critical, closed under repairable direct products, and linear-diagrammatically $(n,m,\ell')$-compatible for integers $n,m \geq 0$, with $n+m >0$, and $\ell' > 0$. We show that $\coll{C}$ is finitely axiomatizable by linear $(n,m,\ell')$-dexrs in three steps:
\begin{enumerate}
	\item We first define a finite set $\dep^\vee$ of dds from $\class{DD}_{n,m,\ell'}^{\ins{S}}$ with at most one atom in the left-hand side of the implication such that, for every $\ins{S}$-structure $I$, $I \in \coll{C}$ iff $I \models \dep^\vee$. This exploits the fact that $\coll{C}$ is $1$-critical (since it is critical) and linear-diagrammatically $(n,m,\ell')$-compatible.
	
	\item We then show that there is a set $\dep^{\exists,=} \subseteq \dep^\vee$ of linear $(n,m,\ell')$-dexrs and $n$-deqrs over $\ins{S}$ such that $\dep^\vee \equiv \dep^{\exists,=}$; in fact, $\dep^{\exists,=}$ is the set of dexrs and deqrs in $\dep^{\vee}$. This uses the fact that $\coll{C}$ is closed under repairable direct products.
	
	\item We finally argue that $\dep^{\exists,=}$ consists only of linear dexrs, which implies that $\coll{C}$ is finitely axiomatizable by linear $(n,m,\ell')$-dexrs. This exploits the fact that $\coll{C}$ is critical.
\end{enumerate}
Note that steps 2 and 3 proceed identically as their counterpart in the proof of Theorem~\ref{the:dexr-refined-characterization} in the main body of the paper. Let us then give a bit more details about step 1.

Let $\dep^{\vee}$ be the set of all dds from $\class{DD}_{n,m,\ell'}^{\ins{S}}$ with at most one atom in the left-hand side of the implication that are satisfied by every structure of $\coll{C}$.
Clearly, $\dep^{\vee}$ is finite (up to variable renaming) since $\dep^{\vee} \subseteq \class{DD}_{n,m,\ell'}^{\ins{S}}$. We show that $\coll{C}$ is precisely the set of $\ins{S}$-structures that satisfy $\dep^{\vee}$.

\begin{lemma}\label{lem:linearization-characterization}
	For every $\ins{S}$-structure $I$, $I \in \coll{C}$ iff $I \models \dep^{\vee}$.
\end{lemma}

\begin{proof}
	The $(\Rightarrow)$ direction holds by construction. We now show the $(\Leftarrow)$ direction. Consider an $\ins{S}$-structure $I$ with $I \models \dep^\vee$. We need to show that $\coll{C}$ is linear-diagrammatically $(n,m,\ell')$-compatible with $I$, which implies $I \in \coll{C}$ since $\coll{C}$ is linear-diagrammatically $(n,m,\ell')$-compatible.
	Consider an arbitrary linear structure $K \subseteq I$ with $\adom{K} = \aadom{K}$ and $|\adom{K}| \leq n$, and an arbitrary $(m,\ell')$-diagram $\Delta_{K,G}^{I}$ of $K$ relative to $I$, where $G \subseteq N_{K,m}^{I}$. We need to show that there exists $J \in \coll{C}$ that satisfies $\Delta_{K,G}^{I}$. To this end, we first establish the following auxiliary claim:
	
	\begin{claim}\label{cla:equivalent-dd-linear}
		There is $\delta \in \class{DD}_{n,m,\ell'}^{\ins{S}}$ that has at most one positive atom in the left-hand side of the implication such that $\delta \equiv \neg \exists \bar x \, \Phi_{K,G}^{I}(\bar x)$.
	\end{claim}
	
	\begin{proof}
		The proof follows the same argument as the proof of Claim~\ref{cla:equivalent-dd}. The key observation is that, since $|\facts{K}| \leq 1$, there exists at most one positive atom in $\Phi_{K,G}^{I}(\bar x)$.
	\end{proof}

	Let $\delta \in \class{DD}_{n,m,\ell'}^{\ins{S}}$ be the dd provided by Claim~\ref{cla:equivalent-dd-linear} such that $\delta \equiv \neg \exists \bar x \, \Phi_{K,G}^{I}(\bar x)$. 
	We claim that $\delta \not\in \dep^\vee$. By contradiction, assume that $\delta \in \dep^\vee$. This implies that $I \models \delta$, which cannot be the case since, by Lemma~\ref{lem:sat-diagram}, $I \models \exists \bar x \, \Phi_{K,G}^{I}(\bar x)$. The fact that $\delta \not\in \dep^\vee$ implies that there exists an $\ins{S}$-structure $L \in \coll{C}$ such that $L \not\models \delta$, which means that $L \models \exists \bar x \, \Phi_{K,G}^{I}(\bar x)$. Therefore, there is an $\ins{S}$-structure $J$ such that $J \simeq L$ and $J \models \Delta_{K,G}^{I}$. Since $\coll{C}$ is closed under isomorphisms, we can conclude that $J \in \coll{C}$, and the claim follows.
\end{proof}

\subsection*{Proof of Lemma~\ref{lem:guarded-compatibility-to-compatibility}}

The proof is along the lines of the proof for Lemma~\ref{lem:linear-compatibility-to-compatibility}.
Consider an arbitrary $\ins{S}$-structure $I$ such that $\coll{C}$ is diagrammatically $(n,m,\ell)$-compatible with $I$. Thus, by definition, $\coll{C}$ is guarded-diagrammatically $(n,m,\ell)$-compatible with $I$. Since $\coll{C}$ is guarded-diagrammatically $(n,m,\ell)$-compatible, we get that $I \in \coll{C}$, which in turn implies that $\coll{C}$ is diagrammatically $(n,m,\ell)$-compatible, as needed.

\subsection*{Proof of Lemma~\ref{lem:guardedization}}

As in the proof of the Linerization Lemma, the proof of the direction $(1) \Rightarrow (2)$ mimics the proof of Lemma~\ref{lem:dexr-locality}, whereas the proof of the direction $(2) \Rightarrow (1)$ mimics the proof of the direction $(2) \Rightarrow (1)$ of Theorem~\ref{the:dexr-refined-characterization}. We leave the details as an easy exercise to the interested reader.

%%%%%%%%%%%%%%%%%

\section{Proofs from Section~\ref{sec:rewritability}}

\subsection*{Proof of Lemma~\ref{lem:boudned-linearization}}

The direction $(2) \Rightarrow (1)$ immediately follows from the Linearization Lemma (direction $(2) \Rightarrow(1)$). We proceed to discuss the interesting direction $(1) \Rightarrow (2)$. By hypothesis, $\coll{C}$ is finitely axiomatizable by $(n,m,\ell)$-dexrs. Thus, there exists a set $\dep$ of $(n,m,\ell)$-dexrs such that, for every $\ins{S}$-structure $J$, $J \in \coll{C}$ iff $J \models \dep$. Consider now an arbitrary $\ins{S}$-structure $I$ such that $\coll{C}$ is linear-diagrammatically $(n,m,\ell')$-compatible with $I$. We need to show that $I \in \coll{C}$, or, equivalently, $I \models \dep$.

Consider an arbitrary dexr $\sigma$ of the usual form $\phi(\bar x,\bar y) \ra \bigvee_{i=1}^{k} \exists \bar z_i \, \psi(\bar x_i,\bar z_i)$ with a non-empty body (the case where $\phi(\bar x,\bar y)$ is empty is treated analogously), and assume that there is a function $h : \bar x \cup \bar y \ra \adom{I}$ such that $h(\phi(\bar x,\bar y)) \subseteq \facts{I}$. We proceed to show that there exists $i \in [k]$ such that $\lambda(\psi_i(\bar x_i,\bar z_i)) \subseteq \facts{I}$ with $\lambda$ being an extension of $h$, which in turn implies that $I \models \sigma$.

Assume that $h(\phi(\bar x,\bar y)) = \{R_1(\bar u_1),\ldots,R_p(\bar u_p)\}$. For each $i \in [p]$, let $K_i$ be the linear structure $(\adom{K_i},R_{i}^{K_i})$, where $\adom{K_i}$ consists of the constants in $\bar u_i$ and $R_{i}^{K_i} = \{\bar u_i\}$, i.e., $\facts{K_i} = \{R_i(\bar u_i)\}$. We proceed to define a suitable $(m,\ell')$-diagram of $K_i$ relative to $I$, for each $i \in [p]$.

A local extension of $h$ for $\psi_i(\bar x_i,\bar z_i)$, where $i \in [k]$, is a function $h' : \bar x_i \cup \bar z_i \ra \adom{I} \cup \bar z_i$ that agrees with $h$ on the variables of $\bar x_i$ and, for each $z \in \bar z_i$, $h'(z) \in h(\bar x) \cup h(\bar y) \cup \{z\}$.
For each $i \in [k]$ and local extension $h'$ of $h$ for $\psi_i(\bar x_i,\bar z_i)$, let $G_i$ be the Gaifman graph of $h'(\psi_i(\bar x_i,\bar z_i))$ restricted to the variables occurring in it, namely the undirected graph $(N_i,E_i)$, where $N_i$ is the set of variables occurring in $h'(\psi_i(\bar x_i,\bar z_i))$ and $(z,z') \in E_i$ iff $z,z'$ are variables that occur together in an atom of $h'(\psi_i(\bar x_i,\bar z_i))$. We further say that a set of atoms $A \subseteq h'(\psi_i(\bar x_i,\bar z_i))$ is connected if the set of variables occurring in $A$ induces a connected subgraph of $G_i$. Moreover, $A$ is maximally connected if (i) it is connected, and (ii) there is no $A' \subseteq h'(\psi_i(\bar x_i,\bar z_i))$ that is connected and $A \subset A'$. For a maximally connected set of atoms $A \subseteq h'(\psi_i(\bar x_i,\bar z_i))$, we denote by $\chi_A$ the sentence $\exists \bar z \bigwedge_{\alpha \in A} \alpha$, where $\bar z$ collects all the variables that occur in $A$. Finally, we define the set $\Gamma$ that collects all the sentences $\chi_A$, where, for $i \in [k]$ and $h'$ a local extension of $h$ for $\psi_i(\bar x_i,\bar z_i)$, the set $A \subseteq h'(\psi_i(\bar x_i,\bar z_i))$ is maximally connected. We can now show the following technical claim via an easy combinatorial argument:

\begin{claim}\label{cla:bound}
	It holds that $|\Gamma| \leq \ell'$.
\end{claim}

\begin{proof}
	Since $\dep$ consists of $(n,m,\ell)$-dexrs, we can conclude that each disjunct in the head of $\sigma$ mentions at most $|\ins{S}| \cdot (n+m)^{\ar{\ins{S}}}$ distinct atoms. Therefore, for each local extension $h'$ of $h$ for $\psi_i(\bar x_i,\bar z_i)$, where $i \in [k]$, we can have at most $|\ins{S}| \cdot (n+m)^{\ar{\ins{S}}}$ maximally connected subsets of $h'(\psi_i(\bar x_i,\bar z_i))$; this is the case where each atom of $h'(\psi_i(\bar x_i,\bar z_i))$ forms a singleton set that is maximally connected. Observe that each local extension $h'$ can assign to a variable $z \in \bar z_i$ at most $n+1$ distinct symbols, and there are at most $m$ distinct variables in $\bar z_i$; thus, we get $(n+1)^m$ different assignments. Consequently, for each $i \in [k]$, $\psi_i(\bar x_i,\bar z_i)$ can give rise to at most 
	\[
	|\ins{S}|\cdot (n+m)^{\ar{S}} \cdot (n+1)^m\ \leq\ |\ins{S}| \cdot (n+m+1)^{m \cdot \ar{\ins{S}}}
	\]
	distinct sentences in $\Gamma$. Since $\sigma$ has a most $\ell$ disjuncts in its head, we get that 
	\[
	|\Gamma|\ \leq\ \ell \cdot |\ins{S}| \cdot (n+m+1)^{m \cdot \ar{\ins{S}}},
	\]
	and the claim follows.
\end{proof}

Now, for each $i \in [p]$, let $\Gamma_i$ be the set of sentences $\gamma \in \Gamma$ that mention (apart from variables) only constants occurring in $R_i(\bar u_i)$ and $I \not\models \gamma$. We define the sentence
\[
\Delta_i\ =\ R_i(\bar u_i)\ \wedge\ \bigwedge_{c,d \in \bar u_i, c \neq d} \neg (c=d)\ \wedge\ \bigwedge_{\gamma \in \Gamma} \neg \gamma.
\]
By construction, each sentence $\gamma \in \Gamma_i$ mentions only constants from $\bar u_i$ and a most $m$ additional variables that are existentially quantified. Therefore, it is not difficult to verify that $\Delta_i$ is an $(m,\ell')$-diagram of $K_i$ relative to $I$. Since $\coll{C}$ is linear-diagrammatically $(n,m,\ell')$-compatible with $I$, for each $i \in [p]$, there exists a structure $J_i \in \coll{C}$ such that $J_i \models \Delta_i$. We assume that, for each $i,j \in [p]$, $\adom{J_i} \cap \adom{J_j} = \bar u_i \cap \bar u_j$, i.e., the active domains of the structures $J_i$ and $J_j$ share only constants occurring in both $\bar u_i$ and $\bar u_j$. This assumption does not affect the generality of the proof since $\coll{C}$ is closed under isomorphisms.

We now define the structure $J^\star$ as the union of the structures $J_1,\ldots,J_p$, i.e., $J^\star$ is such that \[
\adom{J^\star} = \bigcup_{i \in [p]} \adom{J_i} \,\,\,\, \text{and} \,\,\,\, \facts{J^\star} = \bigcup_{i \in [p]} \facts{J_i}.
\]
Since, by hypothesis, $\coll{C}$ is finitely axiomatizable by linear dexrs, it is not difficult to show that $\coll{C}$ is closed under unions, which implies $J^\star \in \coll{C}$. Clearly, $h(\phi(\bar x,\bar y)) \subseteq \facts{J^\star}$, and thus, there is an extension $\mu$ of $h$ such that $\mu(\psi_i(\bar x_i,\bar z_i)) \subseteq \facts{J^\star}$, for some $i \in [k]$.
With $\bar z = \bar z_{i}^{1},\ldots,\bar z_{i}^{m_i}$, where $m_i \leq m$, assume that $\mu(z_{i}^{1}),\ldots,\mu(z_{i}^{m'_{i}}) \in h(\bar x) \cup h(\bar y)$, for some $m'_{i} \leq m_i$, and $\mu(z_{i}^{m'_{i}+1}),\ldots,\mu(z_{i}^{m_{i}}) \not\in h(\bar x) \cup h(\bar y)$.
%
%Then, the extension $\mu'$ of $\mu$ such that $\mu'(z_{i}^{j}) = \mu(z_{i}^{j})$ for each $j \in [m'_{i}]$, and $\mu'(z_{i}^{j}) = z_{i}^{j}$ for each $j \in \{m'_{i}+1,\ldots,m_i\}$, is a local extension of $h$ for $\psi_i(\bar x_i,\bar z_i)$.
%
Then, the mapping $\mu'$ such that $\mu'(z_{i}^{j}) = \mu(z_{i}^{j})$
for each $j \in [m'_{i}]$, and $\mu'(z_{i}^{j}) = z_{i}^{j}$ for each
$j \in \{m'_{i}+1,\ldots,m_i\}$, is a local extension of $h$ for
$\psi_i(\bar x_i,\bar z_i)$.

Let $A_1, \ldots, A_q$ be the maximally connected subsets of $\mu'(\psi_i(\bar x_i,\bar z_i))$ and let $\eta_j$, for each $j \in [q]$, be the formula $\exists \bar z_j \bigwedge_{\alpha \in A_j} \alpha$, where
$\bar z_j$ collects all the variables occurring in $A_j$. We proceed to show the following technical claim.

\begin{claim}\label{cl:lin-sat-union}
	For each $j \in [q]$, there is $k \in [p]$ with $J_k \models \eta_j$.
\end{claim}

\begin{proof}
	By construction, $\mu(\eta_j) \subseteq \mu(\psi_i(\bar x_i,\bar z_i))$. Since
	$\mu(\psi_i(\bar x_i,\bar z_i)) \subseteq \facts{J^\star}$, we can conclude that $\mu(\eta_j) \subseteq \facts{J^\star}$ as well. To conclude the proof, we simply need to show that there exists $J_k$, where $k \in [p]$, such that $\mu(\eta_j) \subseteq \facts{J_k}$.
	If $\mu(\eta_j)$ is a singleton, then the claim follows trivially from the fact that $\facts{J^\star} = \bigcup_{k \in [p]} \facts{J_k}$.
	Assume now that $\mu(\eta_j)$ consists of more than one atoms. Since $\mu(z_{i}^{j}) \not\in h(\bar x) \cup h(\bar y)$, by construction, for each such $\mu(z_{i}^{j})$, there exists only one $k \in [p]$
	such that $\mu(z_{i}^{j}) \in \adom{J_k}$. Moreover, observe that $\eta_j$ defines a connected Gaifman graph over the variables $z_{i}^{j}$ with $j \in \{m'_{i}+1,\ldots,m_i\}$. Thus, every $\alpha \in \mu(\eta_j)$ shares at least one constant $c$ with
	another $\beta \in \mu(\eta_j)$ for which there exists a unique $k \in [p]$ such that $c \in \adom{J_k}$. We can conclude that $\alpha$ and $\beta$ belong to $\facts{J_k}$, and the claim follows.
\end{proof}

Consider now $\eta_j$ for $j \in [q]$, and let $J_k$ for $k \in [p]$ be such that $J_k \models \eta_j$. Recall that $J_k$ satisfies $\Delta_k$ and, by construction, $\adom{J_k}$ contains constants from $h(\bar x) \cup h(\bar y)$ iff they occur in $R_k(\bar u_k)$. We can conclude that $\eta_j$ only mentions constants
occurring in $R_k(\bar u_k)$.

Next, we proceed to prove that $I \models \eta_j$. By contradiction, assume that $I \not\models \eta_j$. Observe that $\eta_j$ is defined from a maximally connected subset of $\mu'(\psi_i(\bar x_i, \bar z_i))$ and, additionally, $\eta_j$ only mentions constants occurring in $R_k(\bar u_k)$. Then, by construction, there exists a conjunction $\neg \gamma$ of $\Delta_k$
such that $\gamma$ and $\eta_j$ are equivalent (up to variable renaming). Thus, we get that $J_k \models \neg \eta_j$, which contradicts the fact that $J_k \models \eta_j$.

Since $I \models \eta_j$, there exists a function
$\lambda_j : \bar z_j \to \adom{I}$ such that
$\lambda_j(\eta_j) \subseteq \facts{I}$. Additionally, observe that, by construction, every pair $\eta_j$ and $\eta_k$, with $j,k \in [q]$ and $j \neq k$, share no variables. Thus, we can define the mapping $\lambda'$ such that $\lambda'(z) = \lambda_j(z)$, if $z$ occurs in $\eta_j$. Finally, we define $\lambda = \lambda' \circ \mu'$. To conclude the proof, we observe that $\lambda$ is an extension of $h$ since $\mu$ is an
extension of $h$ and $\mu'$ redefines only the value of some of the variables in $\bar z_i$. Additionally, since
$\mu'(\psi_i(\bar x_i, \bar z_i)) = \cup_{j=1}^q A_j$ and
$\lambda_j(\eta_j) \subseteq \facts{I}$, we get that $\lambda(\psi_i(\bar x_i, \bar z_i)) \subseteq \facts{I}$, and the claim follows.